\documentclass[journal]{IEEEtran}
\usepackage{amsmath,amsfonts}
\usepackage[backend=biber,style=numeric, maxnames=99, sorting=none, isbn=false]{biblatex}
\usepackage{booktabs}
\usepackage{colortbl} 
\usepackage{multirow}
\usepackage[most]{tcolorbox}
\usepackage{dsfont}
\usepackage{command_MU}
\usepackage{amsmath}
\usepackage{scalerel}
\usepackage{algorithm}
\usepackage{algpseudocode}
\usepackage{graphicx}
\usepackage{tikz}
\usetikzlibrary{positioning}
\usetikzlibrary{calc}
\usetikzlibrary{shapes,arrows}
\usepackage{algpseudocode}
\usepackage{algorithm}
\usepackage{stmaryrd}
\usepackage{tikz}
\usepackage[outline]{contour}
\newcommand{\fixme}[1]{\textcolor{black}{#1}}

\usepackage{hyperref}
\usepackage{placeins}
\usepackage[scaled=.7]{beramono}

\usepackage[most]{tcolorbox}
\definecolor{softblue}{rgb}{0.0627, 0.3333, 0.6039} 
\hypersetup{
    colorlinks=true,        
    linkcolor=softblue,        
    citecolor=softblue,        
    urlcolor=softblue,         
    pdfborder={0 0 0}      
}

\usepackage{enumitem}

\DeclareMathOperator*{\argmin}{arg\,min}

\usepackage{enumitem}

\title{Generalized Ray Tracing with Basis Functions\\ for Tomographic Projections}

\author{Youssef Haouchat\textsuperscript{1}, Sepand Kashani\textsuperscript{2}, Philippe Th\'evenaz\textsuperscript{1},~\IEEEmembership{Fellow,~IEEE} and Michael Unser\textsuperscript{1},~\IEEEmembership{Fellow,~IEEE}
\thanks{
\noindent\textsuperscript{1} Biomedical Imaging Group, EPFL, Switzerland \\
\indent\textsuperscript{2} Center for Imaging, EPFL, Switzerland}}

\AtEveryBibitem{\clearfield{doi} \clearfield{url} \clearfield{eprint}}

\begin{document}

\newcommand\submittedtext{%
  \footnotesize This work has been submitted to the IEEE for possible publication. Copyright may be transferred without notice, after which this version may no longer be accessible.}

\newcommand\submittednotice{%
  \begin{tikzpicture}[remember picture,overlay]
    \node[anchor=south,yshift=10pt] at (current page.south) {\fbox{\parbox{\dimexpr0.65\textwidth-\fboxsep-\fboxrule\relax}{\submittedtext}}};
  \end{tikzpicture}%
}
 \maketitle

\begin{abstract}

This work aims at the precise and efficient computation of the x-ray projection of an image represented by a linear combination of general shifted basis functions that typically overlap. We achieve this with a suitable adaptation of ray tracing, which is one of the most efficient methods to compute line integrals. In our work, the cases in which the image is expressed as a spline are of particular relevance. The proposed implementation is applicable to any projection geometry as it computes the forward and backward operators over a collection of arbitrary lines. We validate our work with experiments in the context of inverse problems for image reconstruction to maximize the image quality for a given resolution of the reconstruction grid. \\
\end{abstract}

\begin{IEEEkeywords}
X-ray, splines, image reconstruction, inverse problems
\end{IEEEkeywords}

\section{Introduction}
In this paper, we propose an exact method to compute the x-ray transform of an image with arbitrary geometry. The accuracy results from the high order of approximation that comes with the representation of data through spline models, combined with the exact computation of the integrals found in x-ray transforms. Our approach turns out to be computationally efficient as well.

\subsection{State of the Art}

The computation of ray-based operators for the x-ray transform often takes advantage of a ray tracer that performs a line integration. Widely used open-source software packages such as the Astra toolbox~\cite{astra}, the TIGRE toolbox \cite{tigre}, or the Reconstruction Toolkit \cite{rtk} employ the Siddon ray tracing algorithm \cite{siddon} or its accelerated variant~\cite{siddonaccelTigre}.  These methods assume the image is piecewise-constant over the cells delimited by a grid, which corresponds to a pixel-based representation. To avoid blocking artifacts or to make an implicitly smoother description of the image, they also implement interpolation-kernel methods \cite{josephkernel} or exploit GPU-based texture managers. For predefined projection geometries, one can also resort to more advanced methods such as footprint-based~\cite{fesslercone}, distance-driven \cite{distancedriven2}, or convolutional \cite{convolutional2} for more precise projection models. The authors of \cite{entezari, mehrsahtv} have studied the impact of a richer representation of the image in the context of x-ray imaging but they do not exploit the computational efficiency of ray tracing as we do in our approach.

\subsection{X-Ray-Transform with Basis Functions}

Let a line on the plane $\R^2$ be described parametrically in terms of $y\in \R$ as the set
\begin{equation}
\label{rayparam}
    \{t\V\theta\ + y\V\theta^{\perp}\in \R^2\ |\ t \in \R\},
\end{equation}

\noindent where $\V\theta = (\cos\theta,\ \sin\theta)\in \mathbb{S}^1(\R)$ is the unit vector directing the line that forms an angle $\theta \in \R$ with the horizontal axis. The vector $\V\theta^\perp = (\sin\theta,\ -\cos\theta) \in \mathbb{S}^1(\R)$ is a unit vector, orthogonal to $\V\theta$, such that $y\V\theta^{\perp}$ can be interpreted as the orthogonal shift by $y$ of the ray relative to the origin. In the context of imaging, $y$ is often taken to correspond to the position of a detector. \\
The x-ray transform \cite{Kak, natterer2001mathematics} of the integrable function \(\text{$f : \mathbb{R}^2 \to \mathbb{R}$}\) corresponds to the collection of all its integrals along such lines. It is expressed in terms of $\theta\in\R$ and $y\in\R$ as

\begin{equation}
    \label{eq:xray}
    \mathcal{P}_{\theta}\{f\}(y) = \int_{\R} f(t \V\theta + y \V\theta^\perp)\ \mathrm{d} t.
\end{equation}

\begin{figure}[!t]
    \centering
    \begin{tikzpicture}
        \node[anchor=south west,inner sep=0] (main2) at (7,0.5) {\includegraphics[width=0.6\linewidth]{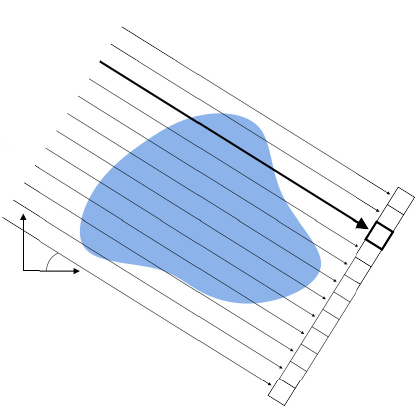}};
        \node[anchor=center, text=black, font=\small ] at (10.55cm,4.5cm) {\contour{white}{object $f$}};
        \node[rotate = 58, anchor=center, text=black, font=\small ] at (7.5cm,4.3cm) {\contour{white}{sources}};
        \node[anchor=center, text=black, font=\small ] at (7.5cm,2.5cm) {$\theta$};
        \node[rotate = 58, anchor=center, text=black, font=\small ] at (11.6cm,1.7cm) {\contour{white}{array of detectors indexed by $y$}};

    \end{tikzpicture}
    \caption{Projection along a line at a specific angle indexed by $\theta$ onto a specific detector indexed by $y$. In this paper, we consider every projection line independently.}
    \label{schemarays}
\end{figure}

\noindent This measurement operator is widely used in tomography to solve inverse problems in modalities such as x-ray scan, positron emission tomography, or cryogenic electron tomography \cite{fessler, epstein2008}. Every detector is indexed by $y$, while every ray is indexed by $\theta$ and $y$, as hinted in Figure~\ref{schemarays}.\\

\noindent The function $f$ in (\ref{eq:xray}) is the model of an image. It is often assumed to consist in the translations of a basis generated from ${\varphi:\R^2 \rightarrow \R}$ placed on a uniform Cartesian grid. For simplicity, and without loss of generality, we choose the stepsize of the grid to be equal to $1$. We also assume $f$ to be compactly supported. Then, we have that
\begin{equation}
\label{eq:def_f}
f = \sum_{\mathbf{k} \in \Omega} c_{\mathbf{k}} \varphi(\cdot - \mathbf{k}),
\end{equation}
\noindent where $\Omega = \{1,\dotsc, N\}\times\{1,\dotsc, N\}$ with $N^2$ the number of basis functions needed to represent $f$, $\mathbf{k}$ is an index on the grid, and $c_{\mathbf{k}}$ is the coefficient of $f$ associated to the shifted basis function $\varphi(\cdot - \mathbf{k})$. Thus, the x-ray projection of a function $f$ parameterized by its coefficients $(c_{\mathbf{k}})_{\mathbf{k}\in \Omega}$ is expressed as

\begin{equation}
\begin{aligned}[b]\label{xrtbasis}
    \mathcal{P}_{\theta}\{f\}(y) &= \sum_{\mathbf{k} \in \Omega} c_{\mathbf{k}}  \mathcal{P}_{\theta}\{\varphi(\cdot - \mathbf{k})\}(y)\\
 &=\sum_{\mathbf{k} \in \Omega} c_{\mathbf{k}}  \mathcal{P}_{\theta}\{\varphi\}(y-\langle \mathbf{k}, \V \theta^\perp \rangle),
\end{aligned}
\end{equation}
with our concern being the efficient and accurate evaluation of (\ref{xrtbasis}) for arbitrary $\theta$ and $y$, given some underlying generator~$\varphi$.\\

\fixme{We summarize below the main notations used throughout the paper. These will be referenced frequently in the sequel.}

\begin{table}[!ht]
\caption{\fixme{Table of Notations}}
\centering
\renewcommand{\arraystretch}{1.2}  
\normalsize
\begin{tabular}{ll}
\toprule
{Symbol} & {Description} \\
\midrule
$\mathbb{S}^1(\R)$ & $\left\{ (\cos\theta, \sin\theta) \in \mathbb{R}^2 \mid \theta \in [0, 2\pi) \right\}$\\
$\V \theta \in \mathbb{S}^1(\R)$ & Direction of the projection\\
$\V \theta^\perp \in \mathbb{S}^1(\R)$ & Direction orthogonal to the projection\\
$L^1(\mathbb{R}^2)$ & $\left\{ f : \mathbb{R}^2 \to \mathbb{R} \ \middle| \ \int_{\mathbb{R}^2} |f(\mathbf{x})|\,\dd \mathbf{x} < \infty \right\}$\\
$\Omega$ & Set of grid indices, denoted $\mathbf{k}\in\R^2$\\
$\left(\theta, y\right)$ & Line parameterization (angle, offset) \\
$\mathcal{P}_\theta\{f\}(y)$ & Line integral of $f$ along $(\theta, y)$\\
$\varphi(\cdot - \mathbf{k})$ & Shifted basis generator centered at $\mathbf{k}$\\
$\varphi_\theta(y)$ & Line integral of $\varphi$ along $(\theta, y)$ \\
$x_+$ & $\max(x,0)$ \\
$\hat{f}$ & Fourier transform of $f \in L^1(\mathbb{R}^2)$ \\
$\mathbf{H}_\varphi$ & Projection operator discretized with $\varphi$ \\
$\delta$ & Dirac distribution\\
\bottomrule
\end{tabular}
\label{tab:notations}
\end{table}

\subsection{Contribution}

Ray-oriented methods such as ray tracing are particularly suitable for the handling of arbitrary geometries which may involve irregular contributions of $y$ and $\theta$ in (\ref{xrtbasis}), thus prohibiting Fourier-based methods in such cases. Their principle (Bresenham's \cite{Bresenham} or Siddon's \cite{jacobs1998fast, siddon, rt} methods) is to cast a ray through a medium composed of simple structures that form a partition of the domain, to compute the intersections between the ray and these structures, and to perform partial integrations over each intersection. Traditional methods correspond to the case when the basis generator $\varphi$ takes a constant value and has a shape that coincides with the cells of the grid. In that case, $\mathcal{P}_{\theta}\{\varphi\}(y)$ is precisely the length of such intersections, to which ray tracing provides direct access.\\ 

\begin{figure}[t]
    \centering
    \begin{tikzpicture}
    \node[anchor=south west,inner sep=0] (main2) at (7,0.5) {\includegraphics[width=1\linewidth]{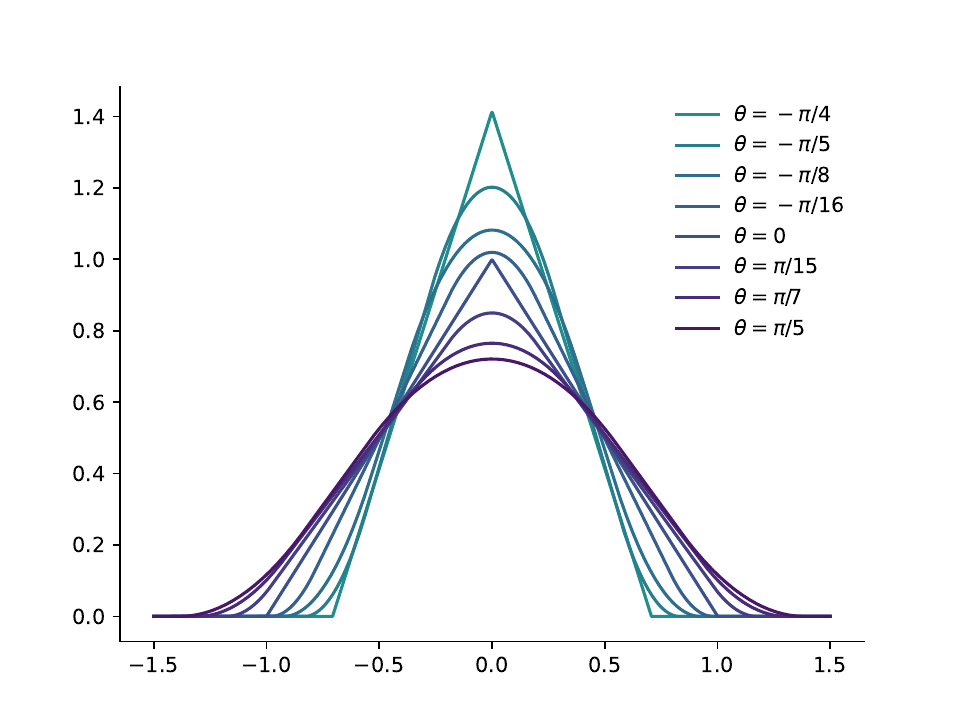}};
    \node[rotate=90, anchor=center, text=black, font=\small ] at (7.2cm,3.65cm) {{$\varphi_\theta(y)$}};
    \node[anchor=center, text=black, font=\small ] at (11.4cm,0.7cm) {{$y$}};
    \end{tikzpicture}
    \caption{Projected bases $y\mapsto \varphi_\theta(y)$. Here, $\varphi$ is the $3$-directional box-spline $\varphi_\theta^{\text{box}}$. The curves are indexed by the projection angle $\theta\in[-\pi/4, \pi/4)$.}
    \label{plotsplines}
\end{figure}

\begin{figure*}[t]
    \centering
\begin{minipage}{0.50\textwidth}
\begin{algorithm}[H]
\caption{Contribution of one basis function to the x-ray}\label{alg:cap}
\begin{algorithmic}
\Require $C, \textcolor{black}{\mathbf{x}_k}$, $\textcolor{black}{\theta}$, $\textcolor{black}{(p,q)}$\\  \Comment{Image, intersection ray/cell, angle, cell index} \\
\State $c = C[\textcolor{black}{p, q}]$ \Comment{Coefficient of the basis function}\\
\State $\textcolor{black}{\mathbf{o}} = (\textcolor{black}{p}+\dfrac{1}{2},\ \textcolor{black}{q}+\dfrac{1}{2})$ \Comment{Position of the center}\\
\State $y_k = \langle(\textcolor{black}{\mathbf{o}} - \textcolor{black}{\mathbf{x}_k}), \textcolor{black}{\V\theta^\perp} \rangle$ \Comment{Offset relative to the center}\\
\\
\Return $P = c \times \varphi_{\textcolor{black}{ \theta}}(y_k)$ \Comment{Integral contribution}
\end{algorithmic}
\end{algorithm}
\end{minipage}%
\hfill
\begin{minipage}{0.45\textwidth}
    \centering
    \begin{tikzpicture}
        \node[anchor=south west,inner sep=0] (main2) at (7,0.5) {\includegraphics[scale = 0.7]{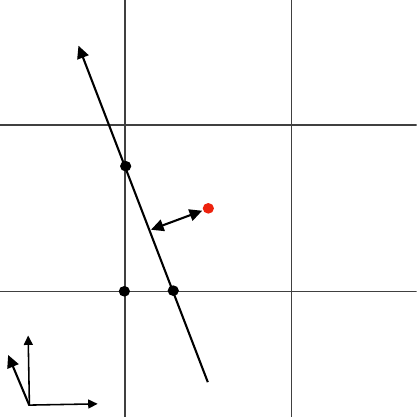}};
        \node[anchor=center, text=black, font=\large ] at (9.05cm,3.1cm) {\contour{white}{$y_k$}};
        \node[anchor=center, text=black, font=\large ] at (9.5cm,1.7cm) {$\mathbf{x}_k$};
        \node[anchor=center, text=black, font=\large ] at (7.9cm,3.5cm) {\contour{white}{$\mathbf{x}_{k+1}$}};
        \node[anchor=center, text=black, font=\large ] at (7.8cm,5.1cm) {\contour{white}{$P$}};
        \contourlength{3.5pt}
        \node[anchor=center, text=black, font=\large ] at (8cm,1.7cm) {$(p,q)$};
        \node[anchor=center, text=red, font=\large ] at (9.75cm,3.1cm) {\contour{white}{$\mathbf{o}$}};
        \node[anchor=center, text=black, font=\large ] at (6.9cm,1.2cm) {$\V\theta$};

    \end{tikzpicture}
    \caption{Relevant quantities in Algorithm \ref{alg:cap}.}
    \label{schemaxrt}
\end{minipage}
\end{figure*}

In this work, we generalize the ray-tracing algorithm to handle any basis function even if the support overflows a cell while remaining exact. The challenge is to combine the {domain decomposition} (a uniform Cartesian grid allowing ray-tracing techniques) with an expansion into basis functions whose support is general and does not coincide with that of the cell  \cite{herman2009fundamentals, unserblu}. Higher-order basis functions allow us to improve image representations, which ultimately leads to better quality.
The contributions of this paper are as follows.\\

\noindent (1) We propose an efficient variant of ray tracing for the exact computation of the x-ray transform and its adjoint, able to handle overlapping basis functions. The algorithm can handle any projection geometry and computes $\mathcal{P}_{\theta_m}\{f\}(y_m)$ for any set of rays parameterized by $(\theta_m, y_m)_{m\leq M}$.\\
\noindent (2) We provide explicit expressions of $\mathcal{P}_{\theta}\{\varphi\}(y)$ for specific basis functions, which we take advantage of in our implementation.\\
\noindent (3) We present image reconstruction experiments that showcase the benefits of higher-order basis functions in tomographic reconstruction problems.\\

\noindent In the sequel, we rely on the shorthand notation $\varphi_\theta$ to denote the x-ray transform at the projection angle $\theta$ of the basis generator $\varphi$, as shown in Figure \ref{plotsplines}, with $y \in \R$ such that
\begin{equation}
\label{eqproj}
    \varphi_\theta : y \mapsto  \mathcal{P}_{\theta}\{\varphi\}(y).
\end{equation}

\section{Generalized x-Ray Projections}

We propose an efficient algorithm to compute the x-ray projection (\ref{xrtbasis}) for an arbitrary ray of angle $\theta$ and offset $y$ using ray tracing. To streamline the exposition, we defer the derivations of closed-form expressions for \(\varphi_\theta(y)\) to Section~\ref{specific}.

\subsection{Ray-Tracing Routine}
\label{section:raytracing}

 Our geometric setup is as follows: the coefficients of $f$ in the chosen basis are stored in an image array $C$ and placed on a two-dimensional uniform Cartesian grid, as sketched in Figure \ref{octagon}. This tiling allows us to efficiently compute the intersection points $\mathbf{x}_{k}$ between rays and cells. In our implementation, we sequentially compute such intersections with a ray-tracing routine by taking advantage of the computer-graphics library Dr.Jit~\cite{Jakob2020DrJit}. Traditional pixel-based approaches output the sum of all the intersection lengths given by $c_k\|\mathbf{x}_{k+1}-\mathbf{x}_{k}\|_2$, where $c_k$ is the value of the crossed pixel. This corresponds to the evaluation of (\ref{xrtbasis}) with $\varphi$ being the rectangle function. 

\subsection{Proposed General Algorithm}

\begin{algorithm}[t]
\caption{\centering X-ray projection with \fixme{overlapping} basis functions via ray tracing}\label{alg:full}
\begin{algorithmic}
\Require $C, \textcolor{black}{\theta}, \textcolor{black}{\mathbf{x}_0}$\\  \Comment{Image, angle, first intersection point of ray tracing}\\

\While{\textit{ray tracing is ongoing}}
\State $\textcolor{black}{\mathbf{x}_{k+1}} = $ Ray-trace$(\textcolor{black}{\mathbf{x}_k}, \textcolor{black}{\theta})$ \Comment{Ray-tracing update}\\

\State $\textcolor{black}{(p,q)} = \left\lfloor\dfrac{\textcolor{black}{\mathbf{x}_k} + \textcolor{black}{\mathbf{x}_{k+1}}}{2}\right\rfloor$ \Comment{Indices of crossed cell}\\

\State $P \gets P\ + $ Algorithm \ref{alg:cap} $(C, \textcolor{black}{\mathbf{x}_k}, \textcolor{black}{\theta}, (\textcolor{black}{p,q}))$\\

\If{$\textcolor{black}{{x}_{k+1,1}} \neq \textcolor{black}{p}$ \textbf{and} $\textcolor{black}{{x}_{k,1}} \neq \textcolor{black}{p}$}
\State $\textcolor{black}{(p_\text{L}, q_\text{L})} = \textcolor{black}{(p-1, q)}$ \Comment{Indices of left cell}\\

\State $P \gets P\ + $ Algorithm \ref{alg:cap} $(C, \textcolor{black}{\mathbf{x}_k}, \textcolor{black}{\theta}, (\textcolor{black}{p_\text{L},q_\text{L}}))$
\EndIf \\

\If{$\textcolor{black}{{x}_{k+1,1}} \neq \textcolor{black}{p+1}$ \textbf{and} $\textcolor{black}{{x}_{k,1}} \neq \textcolor{black}{p+1}$}

\State $\textcolor{black}{(p_\text{R},q_\text{R})} = \textcolor{black}{(p+1, q)}$ \Comment{Indices of right cell}\\
\State $P \gets P\ + $ Algorithm \ref{alg:cap} $(C, \textcolor{black}{\mathbf{x}_k}, \textcolor{black}{\theta}, (\textcolor{black}{p_\text{R},q_\text{R}}))$
\EndIf

\State $\textcolor{black}{k} \gets \textcolor{black}{k+1}$

\EndWhile

\Return $\mathcal{P}_{\textcolor{black}{\theta}}\{f\}(y) = P$

\end{algorithmic}
\end{algorithm} 

In this work, we consider basis functions whose support extends over more than one cell. These generalized bases overlap and require us to revisit the traditional procedure. In order to compute (\ref{xrtbasis}) efficiently with overlapping basis functions, we embed the ray tracer in Algorithm~\ref{alg:full} that also considers some non-intersected cells. The ray-tracing routine discussed in Section \ref{section:raytracing} efficiently computes the update ``$\mathbf{x}_{k+1}~=~$~{Ray-trace}\((\mathbf{x}_k, \theta)\)'', where the index \( k \) refers to the ray-tracing step. At each step $k$, as we cross an intersected cell, we employ Algorithm \ref{alg:cap} (illustrated in Figure \ref{schemaxrt}) to compute the contribution of the basis function centered at this very cell.  
Then, the key point of our method is to complement this {contribution} with that {of the basis functions centered at neighboring cells}.\\

We illustrate our approach in Figure~\ref{octagon}. While we evaluate the contribution of the basis functions centered at the cells that are crossed by the ray (shaded, identified via the ray-tracing routine), we also evaluate the contribution of the surrounding neighbors centered at the cells that are not crossed but still contribute to the line integral (hatched). The number of neighbors to evaluate on the left and right side at each ray tracing step is $\lceil \sqrt{2}R -1\rceil$ according to Lemma~\ref{lemma}, which is $1$ in this case. Two indices of basis functions are emphasized in Figure~\ref{octagon}: $\mathbf{k}_1$, which corresponds to a basis function that contributes to the line integral; and $\mathbf{k}_2$, which corresponds to a basis function that does not.
 In the specific case where $\varphi$ is a rectangle function that matches the shape of a cell, there is no neighbor to consider, and \fixme{Algorithm~2 reduces to Siddon's method}.\\

\contourlength{1pt}
\begin{figure}[t]
    \centering
        \begin{tikzpicture}
        \node[anchor=south west,inner sep=0] (main2) at (7,0.5) {\includegraphics[trim = 3cm 0cm 3cm 0cm, clip, width = \linewidth]{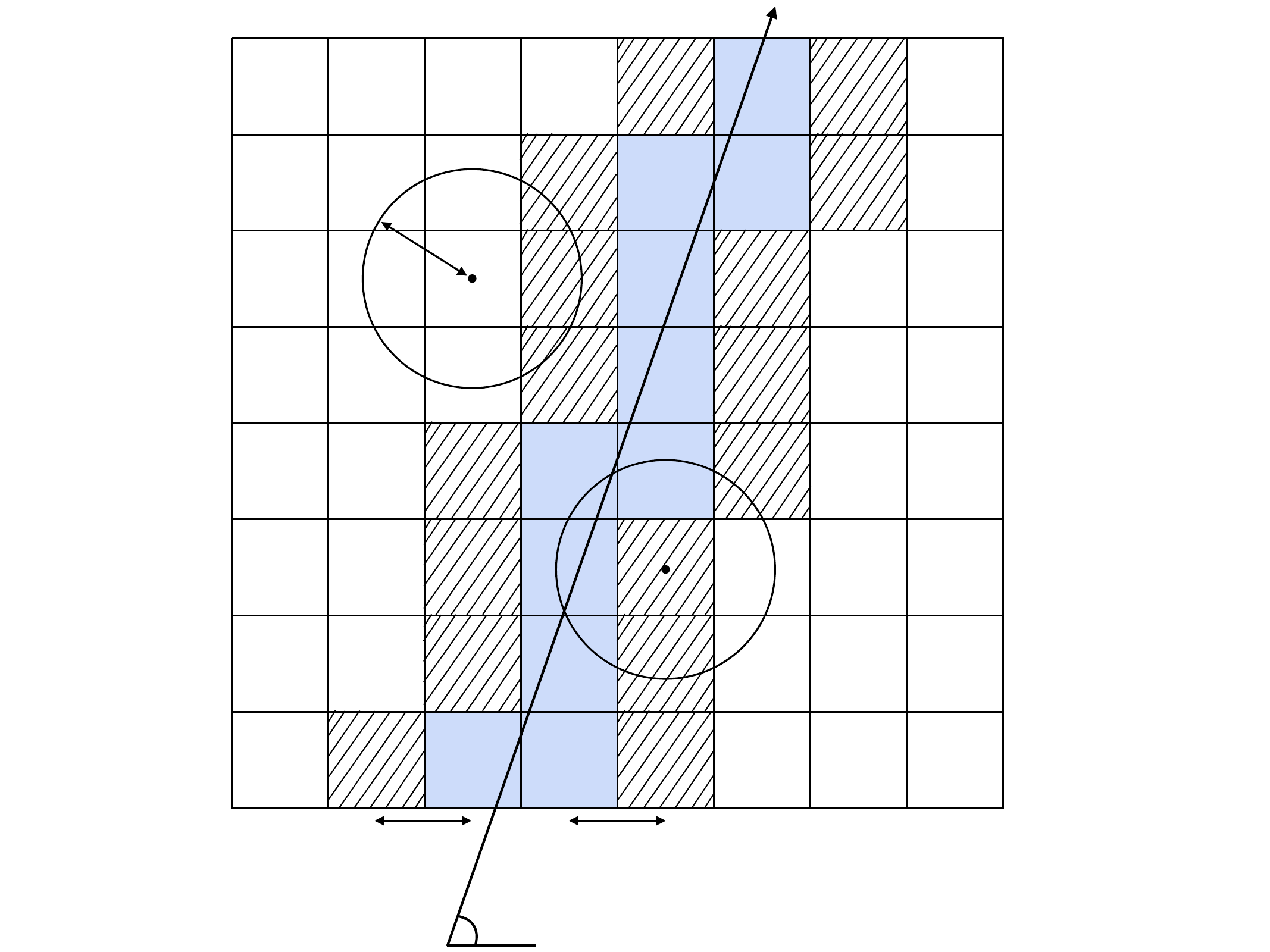}};
        \node[anchor=center, rotate=0, text=black, font=\normalsize ] at (12.7cm,8.6cm) {\contour{white}{$\varphi_{\theta_m}(y_m)$}};
        \node[anchor=center, rotate=0, text=black, font=\normalsize ] at (10.3cm,.8cm) {$\theta_m$};
        \node[anchor=center, rotate=0, text=black, font=\small] at (11.1cm,1.3cm) {$\lceil \sqrt{2}R -1\rceil$};
        \node[anchor=center, rotate=0, text=black, font=\small] at (9.4cm,6cm) {$R$};
        \node[anchor=center, rotate=0, text=black, font=\scriptsize] at (10.1cm,5.75cm) {$\mathbf{k}_2$};
        \node[anchor=center, rotate=0, text=black, font=\scriptsize] at (11.7cm,3.4cm) {\contour{white}{$\mathbf{k}_1$}};
        \end{tikzpicture}
    \caption{Ray tracing when the basis generator has support included in a disk of \fixme{radius} $R$.}
    \label{octagon}
\end{figure}
\contourlength{2pt}

In Algorithm \ref{alg:full}, we only take into account the case where \(|\sin\theta|~>~1/\sqrt{2}\) as in Figure~\ref{octagon}, namely, the ``mainly vertical'' rays where the left- and right-side neighbors are considered. For ``mainly horizontal'' rays (\(|\cos\theta| \geq 1/\sqrt{2}\)), upper and lower neighbors are considered instead. This case distinction on the value of $\theta$ is also part of traditional ray-tracing algorithms.

Algorithm~\ref{alg:full} only includes the case where we consider immediate neighbors, but we have also implemented a version of the approach that is extended to larger basis functions by considering additional neighbors. Lemma~\ref{lemma} provides the number of neighbors to consider when one can bound the shape of the support by either a disk or an octagon. The computational cost of the algorithm depends on this number of neighbors, as discussed in Section~\ref{runtime}.

\begin{lemma}
\label{lemma}
Let $\varphi\in L^1(\R^2)$ be a compactly supported basis generator. Let $
f~=~\sum_{\mathbf{k}} c_{\mathbf{k}} \varphi(\cdot - \mathbf{k})$
and consider a line integral along a ``mainly vertical'' (horizontal, respectively) line. The basis functions that contribute to this line integral are among those whose centers lie within the grid cells that are at most $K$ \textit{horizontal} (\textit{vertical}, respectively) neighbors away from the intersected cells.\\
\noindent (1) If the support of $\varphi$ is included in a disk of radius $R$, then
\begin{equation}
    K = \lceil \sqrt{2}R - 1 \rceil.
\end{equation}
(2) If the support of $\varphi$ is included in an octagonal  shape of girth $L$ as is Figure~\ref{octagon_shape}, then
\begin{equation}
    K=\lfloor L/2 \rfloor.
\end{equation}

\end{lemma}

\begin{figure}[t!]
    \centering
    \begin{tikzpicture}
        \node[anchor=south west,inner sep=0] (main2) at (7,0.5) {\includegraphics[scale = 0.7]{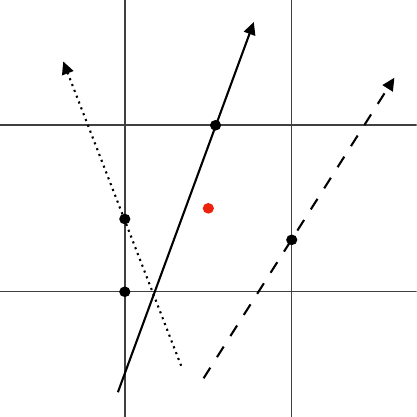}};
        \node[anchor=center, text=black, font=\large ] at (11.1cm,2.5cm) {\contour{white}{$\mathbf{x}_{k+1}$}};
        \node[anchor=center, text=black, font=\large ] at (7.8cm,2.8cm) {\contour{white}{$\mathbf{x}_{k+1}$}};
        \node[anchor=center, text=black, font=\large ] at (9.1cm,4.2cm) {\contour{white}{$\mathbf{x}_{k+1}$}};
        \node[anchor=center, text=black, font=\large ] at (8cm,1.7cm) {$(p,q)$};
\end{tikzpicture}
    \caption{Conditions for neighbor evaluation: one does {not} evaluate the contribution of the left- (right-, respectively) neighbor coefficient for the dotted (dashed, respectively) ray.}
    \label{fig:condition}
    \vspace{-0.3cm}
\end{figure}

\noindent The proof of Lemma~\ref{lemma} is provided in Appendix~\ref{appendixproof}.\\

In Algorithm~\ref{alg:full} and its variants with additional neighbors, immediate neighbors must be evaluated under a specific condition. Indeed, let \((p,q)\) be the index of the crossed cell and $\mathbf{x}_{k}=({x}_{k,1}, {x}_{k,2})$ the entry line-cell intersection point at iteration $k$ of ray tracing. The determination of the contribution is as follows:

\begin{itemize}
    \item If (i) \({x}_{k,1}\) equals \(p\) or (ii) \({x}_{k+1,1}\) equals \(p\), then the left neighbor is not evaluated. It is the case of the dotted ray of Figure \ref{fig:condition}.
    \item If (iii) \({x}_{k,1}\) equals \(p+1\) or (iv) \({x}_{k+1,1}\) equals \(p+1\), then the right neighbor is not evaluated. It is the case of the dashed ray of Figure \ref{fig:condition}.
\end{itemize}

\noindent One should \textit{not} evaluate the contribution of a neighboring cell (I) if the ray just crossed this very cell at iteration $(k-1)$ or (II) if the ray will cross this very cell at iteration $(k+1)$. For ``mainly horizontal'' rays, the conditions concern $x_{k,2}$ and the vertical cell index $q$.\\

\noindent The outcome of Algorithm \ref{alg:full} is the x-ray projection of a linear combination of translations of a basis generator defined on a regular grid for a single ray parameterized by $(\theta, y)$. Instead of $y$, we equivalently take $\mathbf{x}_0$ as an input in the algorithm---it is the first intersection point between the ray and the grid. Our implementation loops in parallel over a collection of arbitrarily parametrized lines. The back-projection operator is the adjoint of the x-ray transform and has been implemented as well, and exactly matched with the forward projection.

\section{Explicit Projections}
An efficient scheme to evaluate $\varphi_\theta$ is needed to make our method practical. This is addressed below with the derivation of closed-form expressions for useful basis generators.

\subsection{Separable Functions}
The x-ray projection of a $2$D separable function is a starting point to apprehend our calculus. Let $\varphi \in L^1(\R^2)$ be a $2$D separable function, in the sense that there exist functions $\phi_1$ and $\phi_2$ in $L^1(\R)$  such that
\begin{equation}
\forall (x_1,x_2)\in\R^2,\ \ \varphi(x_1, x_2) = \phi_1(x_1)\phi_2(x_2).    
\end{equation}

\noindent For the moment, consider only two directions and let $(\mathbf{e}_1, \mathbf{e}_2)$ be the Cartesian basis. Then, Proposition \ref{prop} with $D=2$ states that, for $y\in\R$,
\begin{equation}
\label{xsa}
\varphi_\theta(y) = \Bigg(\ \dfrac{1}{|\cos\theta|}\phi_1\Big(\dfrac{\cdot}{\cos\theta}\Big)\ast \dfrac{1}{|\sin\theta|}\phi_2\Big(\dfrac{\cdot}{\sin\theta}\Big) \ \Bigg)(y).
\end{equation}
\fixme{Let $\hat{\phi}_d$ denote the Fourier transform of $\phi_d$}. If $\hat{\phi}_d(0)=1$, then we can replace the term $\dfrac{1}{|a|}\phi_d\Big(\dfrac{\cdot}{a}\Big)$ by a Dirac distribution $\delta$ when $a = 0$. We manage here to write the x-ray transform of $\varphi$ as a simple convolution, which will lead to explicit expressions of the latter in specific cases (Section \ref{specific}) and to efficient computations as a consequence.

\subsection{Generalization with $D$ Directions}

 Consider now the case when the Fourier transform of the basis generator $\varphi \in L^1(\R^2)$ can be expressed as a product of $D$ terms instead of just two. Let $\hat{f}$ denote the Fourier transform of some function $f$, either in $L^1(\R^2)$ or in $ L^1(\R)$. \\
Then, for any projection angle, it is possible to express the x-ray projection of $\varphi$ as a $1$D convolution of $D$ rescaled atoms.
\begin{proposition}
\label{prop}
Let $\varphi \in L^1(\R^2)$ be such that its Fourier transform $\hat{\varphi}:\R^2 \mapsto \mathbb{C}$ can be written as
\begin{equation}
\label{prod}
\forall \V\xi\in\R^2,\ \  \hat{\varphi}(\V\xi) = \prod_{d=1}^D \hat{\phi}_d(\langle\V\xi, \mathbf{u}_d \rangle ),
\end{equation}
where $\hat{\phi}_d:\R\mapsto\mathbb{C}$ is the Fourier transform of the function $\phi_d:\R\mapsto\R$, and $\mathbf{u}_d\in\R^2$ is a corresponding direction.\\
\noindent Then, we have that the x-ray transform of $\varphi$ along the projection direction $\V\theta\in\mathbb{S}^1(\R)$ is
\begin{equation}
\label{conv}
    \forall y\in\R, \ \ \ \ \varphi_\theta(y) = \Big(\ \tilde\phi_1 \ast\ \cdots\ \ast \tilde\phi_D\ \Big)(y),
\end{equation}
with
\begin{equation}
\tilde\phi_d =
\begin{cases}
 \dfrac{1}{|\langle \V\theta, \mathbf{u}_d\rangle|}\phi_d\Big(\dfrac{\cdot}{\langle \V\theta, \mathbf{u}_d\rangle}\Big), \quad & \langle \V\theta, \mathbf{u}_d\rangle \neq 0 \vspace{0.15cm}\\
\delta, \quad & \langle \V\theta, \mathbf{u}_d\rangle = 0.
\end{cases}
\end{equation}

\end{proposition}

The proof is provided in Appendix~\ref{appendixproof}.

\subsection{Derivation for Specific Basis Functions}
\label{specific}
The authors of \cite{horbelt} present a specific case of Proposition~\ref{prop} when $D=2$ and $\phi_d$ is a B-spline. \fixme{Likewise, choosing $\phi_1,\dotsc,\phi_D$ as rectangles recovers the case of $D$-directional box-splines as in \cite{entezari}. \fixme{In contrast, our formulation in Proposition~\ref{prop} is more general: it applies to any function that is separable in the Fourier domain along arbitrary directions.} In this section, we present cases of particular interest, where we now derive some novel explicit expressions of $\varphi_\theta$ for box-splines and \mbox{B-splines}. These expressions are not found in prior work and allow for the efficient computation of $\varphi_\theta(y)$ in Algorithm~\ref{alg:cap}}.  \\

\noindent Although box-splines are not separable in space, they can be factorized along some directions in the Fourier domain. 
We now derive a closed-form expression of the \mbox{$3$-directional} \mbox{box-spline} with directions $(\mathbf{u}_1, \mathbf{u}_2, \mathbf{u}_3)$, where $(\mathbf{u}_1, \mathbf{u}_2)~=~(\mathbf{e}_1, \mathbf{e}_2)$ forms the Cartesian basis of $\R^2$, and $\mathbf{u}_3 = \mathbf{e}_1 + \mathbf{e}_2$. If $\phi_1,\phi_2,\phi_3$ are centered unit rectangles, then (\ref{prod}) holds true with $D=3$ and
\begin{equation}
    \forall \V\xi\in\R^2,\ \  \widehat{\varphi^{\text{box}}}(\V\xi) = \prod_{d=1}^3  \text{sinc}\Big( \dfrac{\langle\V\xi, \mathbf{u}_d \rangle }{2\pi}\Big).
\end{equation}
\noindent Let $\V \alpha \in \R^6$ denote the vetor
\begin{equation}
\begin{aligned}[b]
\label{alpha}
    \V \alpha = &\Big(0,\, \sin\theta,2\sin\theta+\cos\theta, \ \cos\theta,\\ &\  \sin\theta+2\cos\theta,\  2(\sin\theta+\cos\theta)\Big).
\end{aligned}
\end{equation}

\noindent Then, by simplifying the convolution product (\ref{conv}) in the case of rectangle functions  as suggested in \cite{entezari}, we are able to express the x-ray transform of this $3$-directional box-spline as
    \begin{equation}
    \label{projbox}
    \varphi_\theta^{\text{box}}(y) = \dfrac{-1}{\alpha_1\alpha_3\alpha_5}\sum_{n=1}^6 (-1)^n\ (y-\alpha_n)_+^2,
    \end{equation}

\noindent where $y\in\R$, $\alpha_n$ is the $n$th component of $\V \alpha$, and $(x)_+~:=~\max(x,0)$. In the case where $\theta\in (0,\pi/4)$, we can perform the expansion of the sum in (\ref{projbox}) and collect the polynomial terms given by
\begin{equation}
\label{projectionbox}
\varphi_\theta^{\text{box}}(y) = \begin{cases}
    0, &  y < \alpha_0 \\[8pt]
    \dfrac{y^2}{\alpha_1 \alpha_3 \alpha_5}, & \alpha_0 \leq y < \alpha_1 \\[8pt]
    \dfrac{2y - \alpha_1}{\alpha_3 \alpha_5}, &  \alpha_1 \leq y < \alpha_5 \\[8pt]
    \dfrac{-y^2 + \alpha_5 y - 1}{\alpha_1 \alpha_3 \alpha_5}, & \alpha_5 \leq y < \alpha_{2} \\[8pt]
    \dfrac{-2y + 4 \alpha_3 + 3 \alpha_1}{\alpha_3 \alpha_5}, &  \alpha_{2} \leq y < \alpha_{4} \\[8pt]
    \dfrac{(y -2 \alpha_1 - 2 \alpha_3)^2}{\alpha_1 \alpha_3 \alpha_5}, &  \alpha_{4} \leq y < \alpha_{5} \\[8pt]
    0, &  y \geq \alpha_{5}.
\end{cases}
\end{equation}
\noindent This expression leads to the family of functions shown in Figure \ref{plotsplines}. We derive similar results for the other angle ranges by leveraging inherent symmetries of the box-splines.

\noindent A particular case of box-splines is the \fixme{case of} separable tensor-product B-splines. \fixme{Let $\beta^n$ denote} the univariate B-spline of degree $n$ \cite{unser1993b}. Then, for all $(x_1, x_2)\in\R^2$,
\begin{equation}
    \varphi(x_1, x_2) = \beta^n(x_1)\beta^n(x_2).
\end{equation}
This corresponds to a $2(n+1)$-directional box-spline with $n+1$ directions along $\mathbf{e}_1$ and $n+1$ directions along $\mathbf{e}_2$. We derived a closed-form expression for its x-ray transform in a similar way to (\ref{projbox}) for $n=2$. It is
\begin{equation}
\begin{aligned}
\varphi_\theta(y) =& \frac{1}{6 \, \alpha_1^2 \alpha_2^2} \Big(
    y_+^3
    - 2 \,(y - \alpha_1)_+^3
    - 2 \,(y - \alpha_2)_+^3\\
    &\mbox{}+(y - 2 \alpha_1)_+^3
    + 4 \,(y - \alpha_1 - \alpha_2)_+^3
    + (y - \fixme{2}\alpha_2)_+^3\\
    &\mbox{} - 2 \,(y - \alpha_3)_+^3
    - 2 \, (y - \alpha_4)_+^3
    + (y - \alpha_5)_+^3
\Big).
\end{aligned}
\end{equation}

\noindent More generally, x-ray projections of any $D$-directional box-spline can be computed exactly since the convolution of dilated rectangle functions is known.

\section{Image Reconstruction}
\subsection{Formulation of the Inverse Problem}

In tomography, the corrupted measurements $p\in L^2(\R^2)$ are x-ray projections of the signal of interest $f\in L^2(\R^2)$ and the corruption $\varepsilon$ is commonly assumed to be an additive white Gaussian noise with variance $\sigma^2$. Our model is
\begin{equation}
    p(\theta, y) = \mathcal{P}_\theta\{f\}(y) + \varepsilon(\theta, y).
\end{equation}
Since real measurements are acquired via discrete sensors and detectors, the measurements $p$ are sampled which yields $\mathbf{p}\in\R^M$ such that
\begin{equation}
    \mathbf{p} = \left(p(\theta_m, y_m)\right)_{m=1}^M ,
\end{equation}
\noindent where the set $\left((\theta_m, y_m)\right)_{m=1}^M$ encodes the acquisition geometry. Each component corresponds to a ray, as in  (\ref{rayparam}), which yields
\begin{equation}
    \mathbf{p} = 
    \begin{bmatrix}
           p_{1} \\
           \vdots \\
           p_{M}
         \end{bmatrix}
         =
         \begin{bmatrix}
           \mathcal{P}_{\theta_1} \{f\}(y_1) \\
           \vdots \\
           \mathcal{P}_{\theta_M} \{f\}(y_M)
         \end{bmatrix}
         +
         \begin{bmatrix}
           \varepsilon_{1} \\
           \vdots \\
           \varepsilon_{M}
         \end{bmatrix}.
\end{equation}

\noindent From this finite number of measurements, we \fixme{seek} to recover a finite number of {coefficients} of $f$ in some basis as in (\ref{xrtbasis}). By invoking the linearity of $\mathcal{P}$, one seeks to recover $\mathbf{c}\in \mathbb{R}^{N^2}$ such that
\begin{equation}
\mathbf{p} = \M H_\varphi \mathbf{c} + \V\varepsilon,
\end{equation}
where $\M H_\varphi \in \mathbb{R}^{M\times N^2}$ is the discretized version of the x-ray operator. More precisely, we have that
\begin{equation}
(\M H_\varphi)_{m,\mathbf{k}} = \mathcal{P}_{\theta_m}\{\varphi(\cdot - \mathbf{k})\}({y_m}) = \varphi_{\theta_m}(\tilde{y}_{m,\mathbf{k}}),
\end{equation}
where the entries depend on the choice of the generator $\varphi$ and $\tilde{y}_{m,\mathbf{k}}=(y_m-\langle \mathbf{k},\V\theta^\perp \rangle)$.
We now face the task of finding the solution of a maximum-likelihood optimization problem to obtain
\begin{equation}
\label{opti}
     {\mathbf{c}^*} = \argmin\limits_{\mathbf{c} \in \R^{N^2}}  \frac{1}{2\sigma^2} \norm{\mathbf{p} - \M H_\varphi \mathbf{c}}_2^2 + \lambda \M R_\varphi(\mathbf{c}),
\end{equation}
\noindent where $\M R_\varphi$ is a regularization term that enforces desirable properties of the solution and depends on the choice of $\varphi$, as studied by the authors of \cite{mehrsahtv, mehrsamri}. The parameter $\lambda\in\R_{\geq 0}$ controls the strength of the regularizer.

The recovery of the coefficients \fixme{${\mathbf{c}^*}$} of the target image from sampled measurements provides access to the continuous-domain signal defined by (\ref{xrtbasis}), which enables a reconstruction at any desired sampling rate. If the basis functions are uniform tiles matching the reconstruction grid (\textit{i.e.}, splines of degree 0, a.k.a.\ ``pixels''), then the underlying continuous-domain image is piecewise-constant, thereby offering no enhancement in quality at a higher sampling rate.
The projection operator $\M H_\varphi$ and its adjoint are implemented in a matrix-free format in accordance with Algorithm~\ref{alg:full}. We demonstrate in Section~\ref{section:xp} the effectiveness of our implementation in the context of inverse problems. 
\begin{figure}[t]
    \centering
    \includegraphics[width=0.8\linewidth]{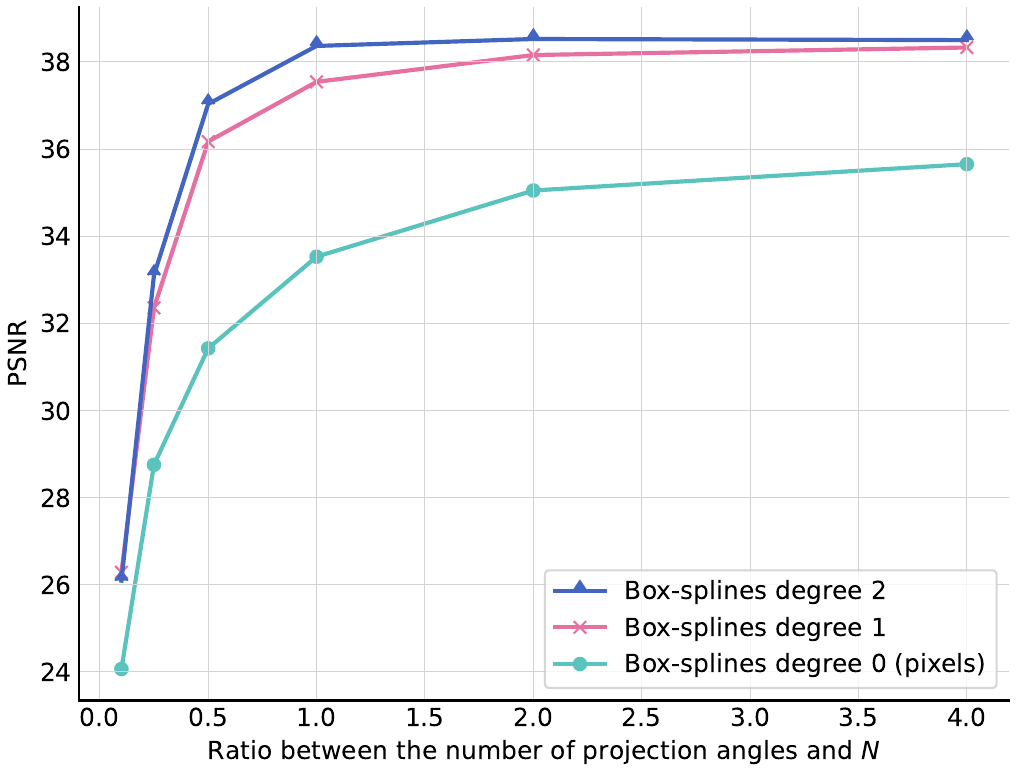}
    \caption{PSNR of the reconstructions in terms of the number of projection angles with a grid of size $N=500$.}
    \label{fig:ratio}                  
\end{figure}
\subsection{Experiments \fixme{with Synthetic Data}}
\label{section:xp}
\begin{figure*}
         \begin{tikzpicture}
         \node[anchor=south west,inner sep=0] (main2) at (3,0) {\includegraphics[clip,keepaspectratio, width=0.97\linewidth]{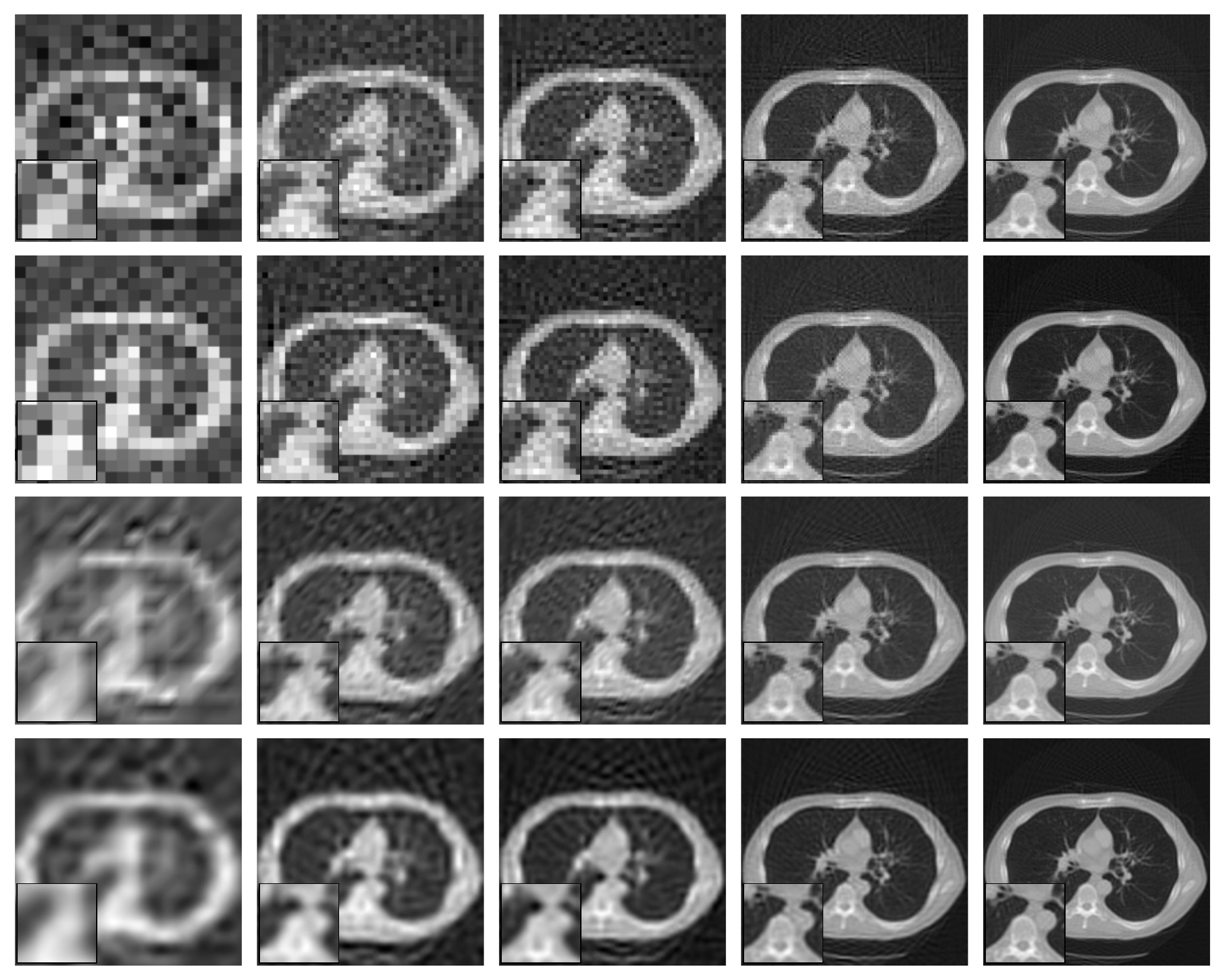}};
         \node[anchor=center, rotate=90, text=black, font=\small ] at (2.8cm,1.8cm) {Box-splines 2};
         \node[anchor=center, rotate=90, text=black, font=\small ] at (2.8cm,5.2cm) {Box-splines 1};
         \node[anchor=center, rotate=90, text=black, font=\small ] at (2.8cm,8.6cm) {Box-splines 0};
         \node[anchor=center, rotate=90, text=black, font=\small ] at (2.8cm,12cm) {Astra \texttt{`fanflat'}};

         \node[anchor=center, text=black, font=\small ] at (4.8cm,13.9cm) {$N_{\text{down}} = 20$};
        \node[anchor=center, text=black, font=\small ] at (8.2cm,13.9cm) {$N_{\text{down}} = 40$};
        \node[anchor=center, text=black, font=\small ] at (11.6cm,13.9cm) {$N_{\text{down}} = 50$};
        \node[anchor=center, text=black, font=\small ] at (15cm,13.9cm) {$N_{\text{down}} = 100$};
        \node[anchor=center, text=black, font=\small ] at (18.4cm, 13.9cm) {$N_{\text{down}} = 150$};
         \end{tikzpicture}
\caption{Cone-beam reconstructions by Algorithm~\ref{alg:full} in terms of grid size and model, with conjugate-gradient descent (30 iterations). Column: reconstructed signals with $N_{\text{down}}^2$ recovered coefficients. Row: Astra reconstruction and box-spline models of degree 0 (pixels), 1, and 2.}     \label{fig:cone}
\end{figure*}

We now conduct experiments that compare the reconstruction performance of spline-based operators with that of traditional pixel-based ones. In order to single out the impact of the discretizations, we focus on scenarios where the problem is sufficiently well posed to avoid the need for regularization (\textit{i.e.} $\lambda=0$ in (\ref{opti})). 
Specifically, we assume a regime in which $M$ is large, with many angles and many offsets, and in which there is no missing cone. For a conventional geometry where the angles and offsets are uniformly spaced, it has been established that this regime occurs when one measures at least $\pi N/2$ angles with at least $N$ offsets, where the reconstruction grid is of size $(N\times N)$. \fixme{This is known as the Crowther criterion}  \cite{rattey, Brooks1978}. Our own experiments in Figure~\ref{fig:ratio} confirm that about \fixme{$\pi N/2$ angles are needed in the traditional pixel-based approach, which we round up to $2N$ in the sequel}. For the basis functions of higher order but larger support proposed in this paper, however, we shall see that even fewer angles are required.\\

\begin{figure*}
        \centering
         \begin{tikzpicture}
         \node[anchor=south west,inner sep=0] (main2) at (3,0) {\includegraphics[trim=0cm 3cm 10cm 0cm,clip,keepaspectratio, width=0.93\linewidth]{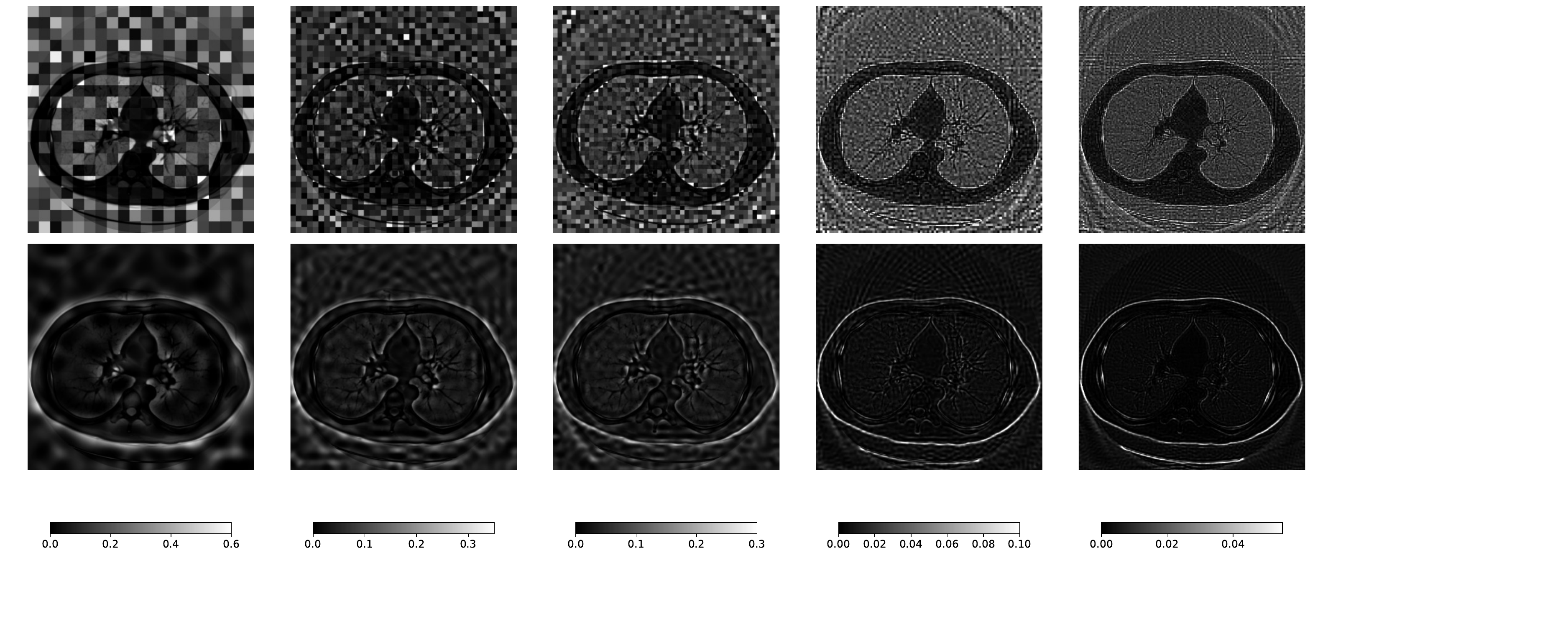}};

         \node[anchor=center, rotate=90, text=black, font=\small ] at (2.8cm,2.5cm) {Box-splines 2 (ours)};
         \node[anchor=center, rotate=90, text=black, font=\small ] at (2.8cm,5.6cm) {Astra \texttt{`fanflat'}};

        \node[anchor=center, text=black, font=\small ] at (4.8cm,7.15cm) {$N_{\text{down}} = 20$};
        \node[anchor=center, text=black, font=\small ] at (8.cm,7.15cm) {$N_{\text{down}} = 40$};
        \node[anchor=center, text=black, font=\small ] at (11.2cm,7.15cm) {$N_{\text{down}} = 50$};
        \node[anchor=center, text=black, font=\small ] at (14.5cm,7.15cm) {$N_{\text{down}} = 100$};
        \node[anchor=center, text=black, font=\small ] at (17.8cm, 7.15) {$N_{\text{down}} = 150$};
        \end{tikzpicture}
\caption{\fixme{Absolute difference images between the ground truth and the most relevant reconstructions from Figure~\ref{fig:cone}.}}     \label{fig:differences}
\end{figure*}

\begin{table*}[!ht]
\caption{\fixme{PSNR and SSIM for the reconstructions of Figure~\ref{fig:cone} in terms of baseline method (Astra and box-splines of degree 0) and our box-splines model (degrees 1 and 2), across grid size}. Bold: best PSNR and SSIM in each row.}
    \normalsize
    \centering
    \begin{tabular}{rcccc}
    \toprule
    {$N_{\text{down}}$} & {Astra \texttt{`fanflat'}} & Box-splines $0$ & Box-splines $1$ & Box-splines $2$ \\
    \midrule
    50   & (19.97, 0.58) & (18.29, 0.54) & (\textbf{20.68}, 0.70) & (20.46, \textbf{0.74}) \\
    100  & (24.81, 0.73) & (23.21, 0.71) & (25.74, 0.84) & \textbf{(26.05, 0.87)} \\
    150  & (27.54, 0.82) & (26.01, 0.81) & (27.97, 0.90) & \textbf{(28.62, 0.93)} \\
    250  & (32.08, 0.90) & (29.78, 0.89) & (32.50, 0.96) & \textbf{(32.85, 0.97)} \\
    375  & (34.90, 0.93) & (32.92, 0.93) & (35.71, 0.98) & \textbf{(35.96, 0.98)} \\
    1000 & (39.46, 0.97) & (38.74, 0.97) & (40.87, 0.99) & \textbf{(41.35, 0.99)} \\
    \bottomrule
\end{tabular}
    \label{tab:psnr_ssim_cone}
\end{table*}

\begin{figure*}[t]
    \centering
    \includegraphics[width=0.4\linewidth]{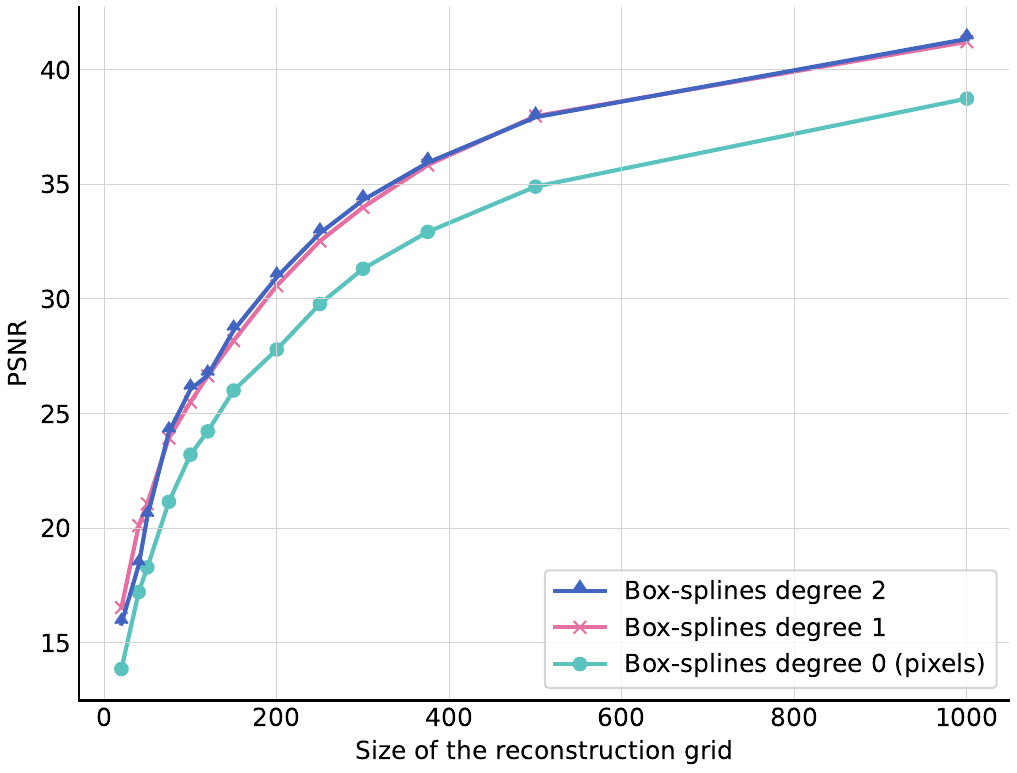}\includegraphics[width=0.4\linewidth]{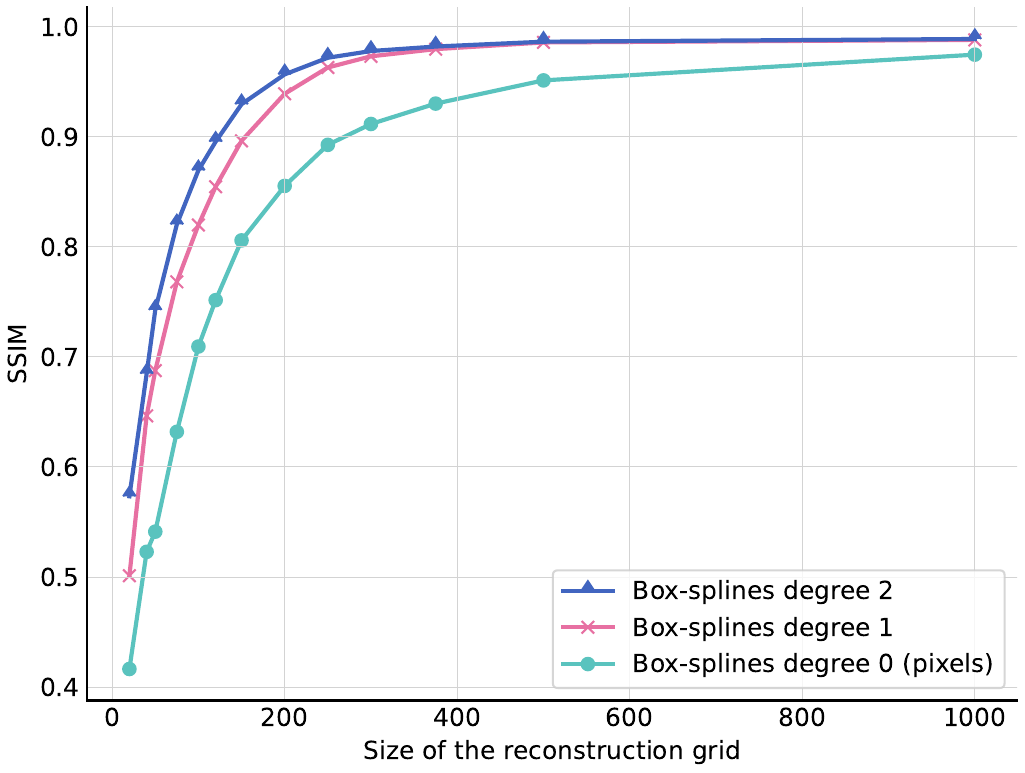}
    \caption{PSNR and SSIM of the reconstructions in terms of grid size for the cone-beam setup.}
    \label{fig:metrics_cone}                  
\end{figure*}
We perform image reconstructions as follows: we choose the ground truth image $\mathbf{x}_{\text{GT}}\in\mathbb{R}^{N^2}$ from the public dataset \cite{dataset} of lung CT medical images, and we upsample them with high resolution $N=3000$ using cubic \mbox{B-splines}. We acquire from $\mathbf{x}_{\text{GT}}$ the projections~$\mathbf{p}\in\mathbb{R}^{2N_{\text{down}}\times N_{\text{down}}}$.
Experiments with parallel-beam and cone-beam projection geometries have been conducted. In order to avoid committing an inverse crime with our model, the data $\mathbf{p}$ are acquired with the Astra-toolbox package under the \texttt{`line'} (\texttt{`line\_fanflat,'} respectively) projector with the \texttt{`parallel'} (\texttt{`fanflat,'} respectively) projection geometry. However\fixme{, even though we add noise to the projections ($\sigma^2 = 10^{-3}$)}, we commit an inverse crime with Astra reconstructions. The solution of the basic optimization problem (\ref{opti}) is obtained via conjugate-gradient descent to recover the coefficients ${\mathbf{c}^*}~\in~\mathbb{R}^{N_{\text{down}}^2}$ in a reconstruction grid of size $(N_{\text{down}}\times N_{\text{down}})$. Using ${\mathbf{c}^*}$, we sample the underlying continuous signal (\ref{xrtbasis}) on a finer grid, yielding ${\mathbf{x}^*}\in \mathbb{R}^{N^2}$.

\begin{table*}[t]
\caption{\fixme{GPU execution times (in milliseconds) for projection and back-projection operations under arbitrary and cone-beam geometries. The operations are performed with $N\times N$ rays on an image of size \mbox{$(N\times N)$}.}}
    \centering
\begin{tabular}{l c c c c | c c c c}
    \toprule
    \multirow{2}{*}{GPU times (ms)} &
    \multicolumn{4}{c|}{Arbitrary geometry} &
    \multicolumn{4}{c}{Cone-beam geometry} \\
    \cmidrule(lr){2-5} \cmidrule(lr){6-9}
    & $N=250$ & $N=500$ & $N=750$ & $N=1000$
    & $N=250$ & $N=500$ & $N=750$ & $N=1000$ \\
    \midrule
    \textbf{ASTRA} \\
    Projection       & 206   & 821   & 1850  & 3350  & 2.60 & 6.20 & 11.9 & 20.0 \\
    Back-projection  & 236   & 1050  & 2890  & 6490  & 1.80 & 4.80 & 9.20 & 17.5 \\
    \midrule
    \textbf{Ours} (degree 0) \\
    Projection       & 0.610 & 2.17  & 6.89  & 16.6  & 0.640 & 2.41 & 6.69 & 16.7 \\
    Back-projection  & 1.48  & 5.16  & 14.8  & 33.9  & 1.14  & 5.90 & 14.8 & 33.9 \\
    \textbf{Ours} (degree 1) \\
    Projection       & 2.23  & 5.08  & 16.8  & 38.7  & 1.50  & 5.26 & 16.9 & 38.8 \\
    Back-projection  & 3.85  & 16.5  & 43.3  & 93.5  & 3.78  & 16.5 & 43.2 & 93.2 \\
    \textbf{Ours} (degree 2) \\
    Projection       & 2.74  & 11.1  & 34.2  & 71.7  & 2.89  & 11.1 & 33.5 & 71.5 \\
    Back-projection  & 4.69  & 18.7  & 51.5  & 111   & 5.05  & 18.7 & 51.3 & 115 \\
    \bottomrule
\end{tabular}
\color{black}
    \label{tab:performance_gpu}
\end{table*}

\begin{table*}[t]
\caption{\fixme{CPU execution times (in milliseconds) for projection and back-projection operations under arbitrary and cone-beam geometries. The operations are performed with $N\times N$ rays on an image of size \mbox{$(N\times N)$}.}}
\centering
\begin{tabular}{l c c c c | c c c c}
    \toprule
    \multirow{2}{*}{CPU times (ms)} &
    \multicolumn{4}{c|}{Arbitrary geometry} &
    \multicolumn{4}{c}{Cone-beam geometry} \\
    \cmidrule(lr){2-5} \cmidrule(lr){6-9}
    & $N=250$ & $N=500$ & $N=750$ & $N=1000$
    & $N=250$ & $N=500$ & $N=750$ & $N=1000$ \\
    \midrule
    \textbf{ASTRA} \\
    Projection       & 296   & 1620  & 4590  & 9820   & 110  & 840  & 2910 & 6860 \\
    Back-projection  & 297   & 1610  & 4590  & 9880   & 110  & 840  & 2890 & 6840 \\
    \midrule
    \textbf{Ours} (degree 0) \\
    Projection       & 1.53  & 8.14  & 32.0  & 77.9   & 1.55  & 8.13  & 30.8  & 72.3 \\
    Back-projection  & 12.6  & 58.2  & 288   & 263    & 13.7  & 59.5  & 254   & 385 \\
    \textbf{Ours} (degree 1) \\
    Projection       & 5.00  & 36.7  & 119   & 264    & 5.11  & 35.9  & 118   & 271 \\
    Back-projection  & 21.5  & 127   & 541   & 1280   & 23.3  & 124   & 484   & 1380 \\
    \textbf{Ours} (degree 2) \\
    Projection       & 10.7  & 127   & 195   & 454    & 8.54  & 128   & 192   & 448 \\
    Back-projection  & 24.7  & 148   & 579   & 1450   & 26.3  & 152   & 564   & 1490 \\
    \bottomrule
\end{tabular}

\label{tab:performance_cpu}
\end{table*}

The reconstruction results are shown in Figure~\ref{fig:cone}  for the cone-beam setup. Except at the coarsest grid sizes, the Astra reconstructions are visually indistinguishable \fixme{from} those of \mbox{box-splines} of degree $0$ (pixels), while \mbox{box-splines} of degree $1$ or $2$ provide a more accurate representation of the target image. In Figure~\ref{fig:metrics_cone}, both the peak signal-to-noise ratio and the structural similarity are reported as curves. \fixme{Numerical values are provided in Table~\ref{tab:psnr_ssim_cone}. These clearly demonstrate that our model, based on higher-degree spline representations, consistently outperforms the traditional pixel-based approach (degree 0), with box-splines of degree~2 achieving the best overall reconstruction quality. The residual images in Figure~\ref{fig:differences} show the absolute difference between the reconstruction and the ground truth, and further illustrate the improvements achieved by higher-degree box-spline models. (We only kept Astra's \texttt{`fanflat'} model and our model of degree 2 for clarity.)}

\begin{figure*}
\centering
\begin{tikzpicture}

  \node[anchor=south,inner sep=0] (img1) at (0,0)
    {\includegraphics[clip,keepaspectratio, width=0.95\linewidth]{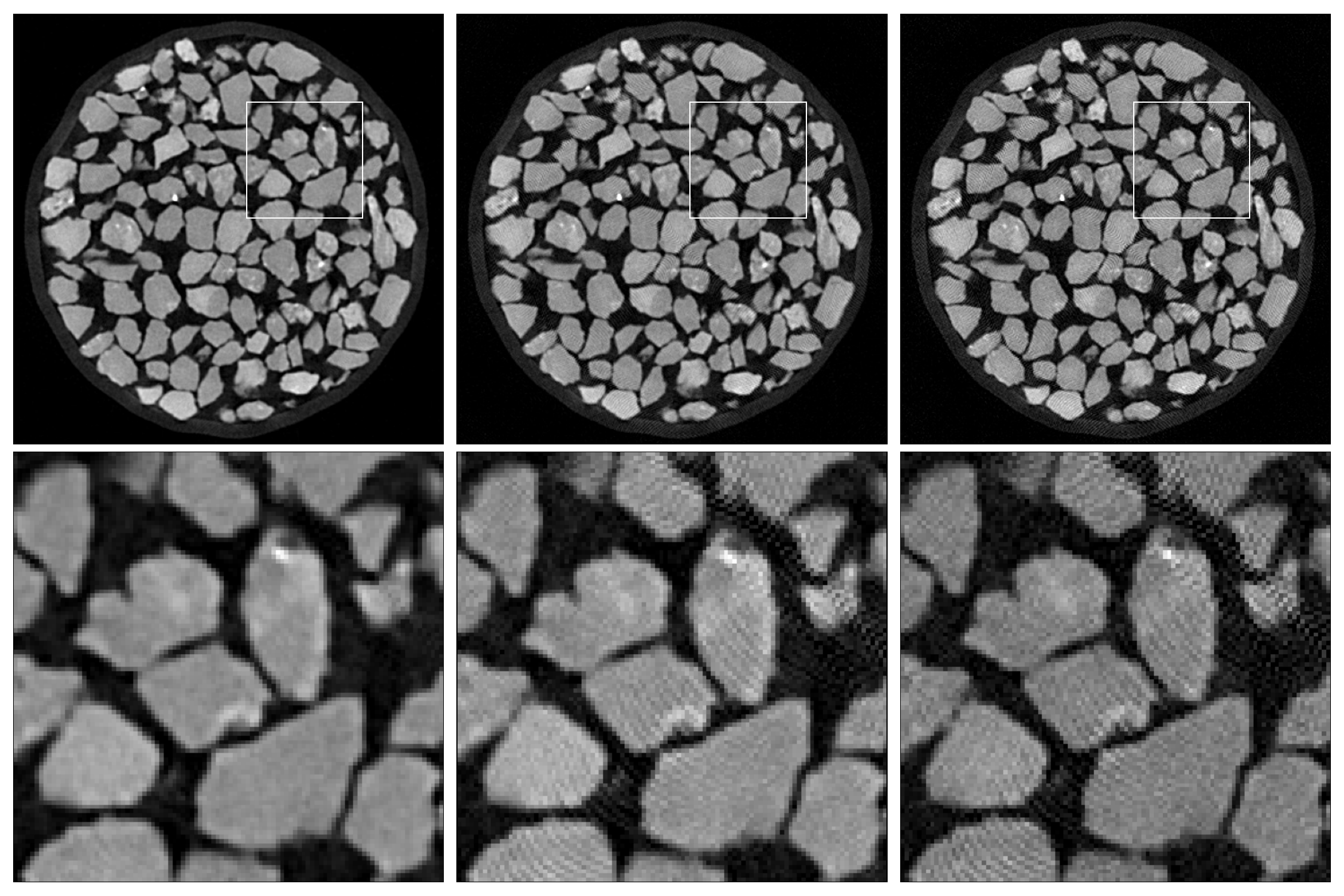}};

  \node[anchor=center, text=black, font=\small ] at (5.4cm,11.25cm) {\fixme{Splines of degree 0 (ours)}};
  \node[anchor=center, text=black, font=\small ] at (0cm,11.25cm) {\fixme{Astra \texttt{`line\_fanflat'} model}};
  \node[anchor=center, text=black, font=\small ] at (-5.4cm,11.25cm) {\fixme{Splines of degree 2 (ours)}};

  \node[anchor=center, rotate=90, text=black, font=\small ] at (-8.5cm,8.4cm) {\fixme{Reconstruction with $800$ views}};
  \node[anchor=center, rotate=90, text=black, font=\small ] at (-8.5cm,3cm) {\fixme{Zoom on the box region}};

  \node[anchor=north,inner sep=0] (img2) at ([yshift=-0.cm]img1.south)
    {\includegraphics[clip,keepaspectratio, width=0.95\linewidth]{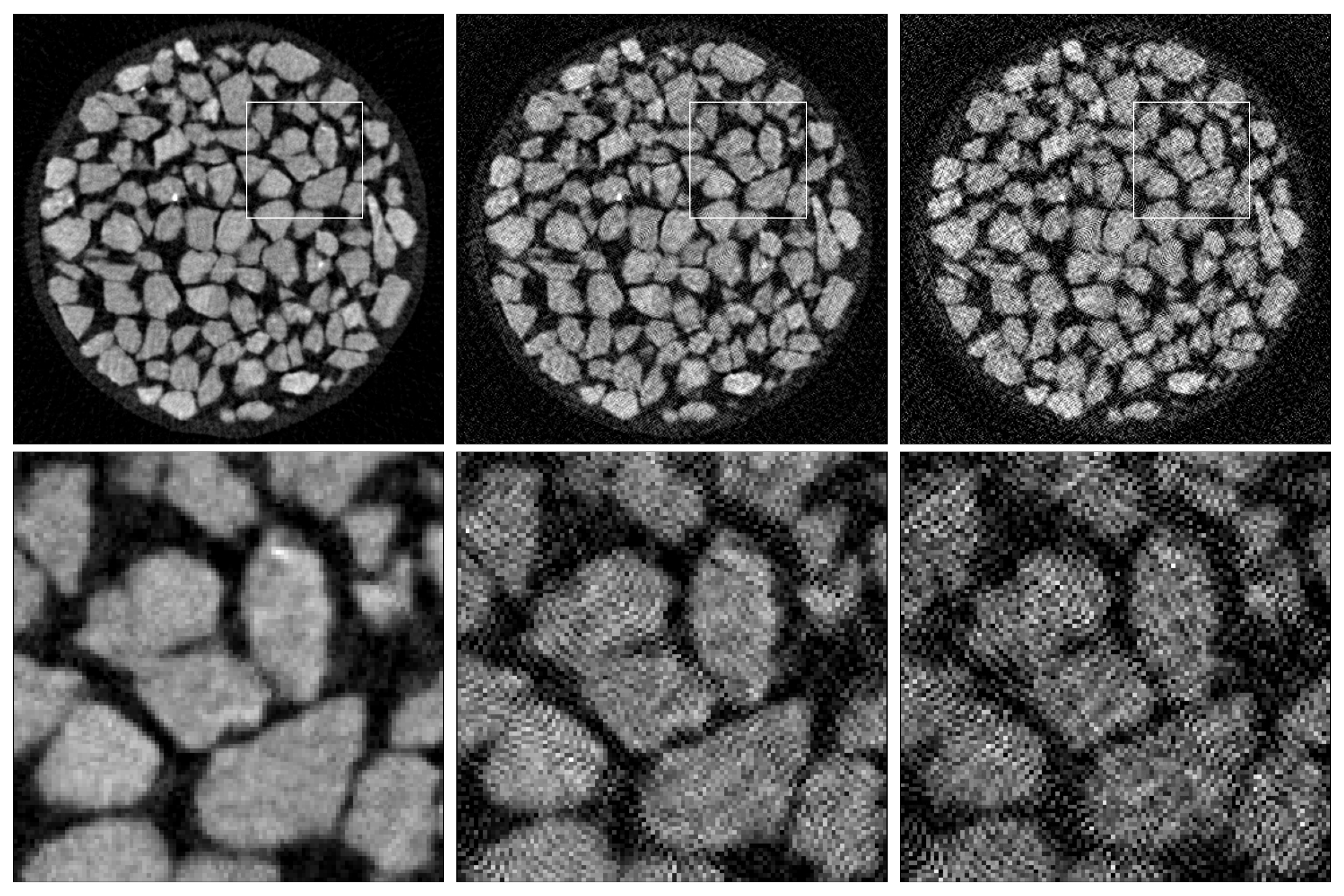}};

    \node[anchor=center, rotate=90, text=black, font=\small ] at (-8.5cm,-3cm) {\fixme{Reconstruction with $200$ views}};
    \node[anchor=center, rotate=90, text=black, font=\small ] at (-8.5cm,-8.4cm) {\fixme{Zoom on the box region}};
    
\end{tikzpicture}
\caption{\fixme{Reconstructions with real tomographic data for different models, with conjugate gradient descent (30 iterations). Astra's \texttt{`line\_fanflat'} model and our pixel-based model (splines of degree 0) are shown as a baseline for comparison with our model of splines of degree 2. First two rows: reconstructions with 800 projection angles. Last two rows: reconstructions with 200 projection angles.}}
\label{fig:experimental_results_full}
\end{figure*}

\subsection{\fixme{Experiments with Real Tomographic Data}}
\label{section:xp_real}
\fixme{We performed image reconstructions with real cone-beam data. The acquisition device is a micro CT-scan with one source and a flat detector screen, aligned with the rotation table of the sample between them. This experimental setup corresponds to a cone-beam geometry. The detector screen is composed of receptors of lateral size $0.127$ mm; the source-to-detector distance is $765.7$ mm and the source-to-object distance is $96.46$ mm. We calibrate the center of rotation of the sample, which is experimentally imperfectly matched with the center of the sample. We achieve this correction through a grid-search algorithm. In this case, the calibration leads to a shift of the center of rotation by $(-0.292)$ mm in the direction orthogonal to the source-detector axis. We have access to $800$ projections and seek to reconstruct the center slice of the volume in a reconstruction grid of size $(728\times728)$. The imaged sample is a collection of rocks inside a cylindrical box.  We demonstrate in Figure~\ref{fig:experimental_results_full} that our model provides superior image quality and artifact reduction compared to pixel-based ray-tracing baselines in two scenarios: with the sinogram of $800$ views; and with a quarter of the measurements where only $200$ views are retained.} \fixme{We compare our model of splines of degree 2 with two baselines: Astra's pixel-based \texttt{`line\_fanflat'} model; and our model of splines of degree 0 (pixels), which yield very similar reconstructions. The reconstruction algorithm remains identical across all cases; only the underlying model differs.}

\subsection{Runtime Evaluation \fixme{and Discussion}}
\label{runtime}
\fixme{Sections \ref{section:xp} and \ref{section:xp_real} evaluated the performance of our algorithm in terms of reconstruction quality. We now evaluate the performance of the algorithm in terms of computation speed.}
\fixme{As expected, the computation times of our algorithm, shown in Tables~\ref{tab:performance_gpu} and \ref{tab:performance_cpu}, does not depend on the projection geometry.} 
\fixme{By contrast, the performance of Astra, which is highly optimized for parallel- and cone-beam geometries, drops significantly for arbitrary geometries. This makes our algorithm particularly well suited for imaging modalities such as positron electron tomography \cite{pet_aleix}, plasma diagnostic imaging, or additive volumetric manufacturing where the geometry is not structured.}
GPU benchmarks were carried out on an NVIDIA RTX A5000 with five warm-up runs. The reported times represent the average of ten executions, and provided by the \texttt{cupyx.benchmark} module. \fixme{Our speedup on CPU is due both to the just-in-time compilation of the ray-tracer and to the fine-grain parallelism that is performed. Each ray is treated independently, which allows us to distribute the computation to multiple cores.}

\fixme{We report in Tables~\ref{tab:performance_gpu} and \ref{tab:performance_cpu} the computation times for both Astra and our spline-based model of degree 0 (pixel baseline), 1 and 2. On average, the runtime increases by a factor of 2.6 for degree 1 and 4.1 for degree 2 relative to the pixel baseline. These averages are computed across all setups of Table~\ref{tab:performance_gpu} and include both projection and back-projection operations on the GPU. The use of box-splines of degree 2 provides the best tradeoff between quality and runtime.} They have the same approximation properties as the separable B-splines of degree 2 (presented in Table~\ref{tab:psnr_ssim_cone_separable} of the Appendix) but have a more compact support that extends over three cells only, against five for the separable B-splines of the same degree. Their support corresponds to that of Figure~\ref{octagon_shape} with $K=1$, an optimal configuration for our algorithm as it corresponds to the largest support for $K$ neighbor evaluations. Splines of degree higher than 2 do not provide a sufficient gain in quality to justify their use, as saturation with respect to the degree is observed in our results (Table~\ref{tab:psnr_ssim_cone} and Figure~\ref{fig:metrics_cone}).

\section{Conclusion}

We presented a framework to compute the x-ray transform operator for 2D signals decomposed into basis functions, specifically, box-splines and separable B-splines. We derived closed-form expressions of their x-ray projections for arbitrary integration lines. Using these analytical derivations, our algorithm can tackle the case of overlapping basis functions with a neighbor-based approach in a scalable and efficient way. Our projector and back-projector form matched adjoint pairs and allow for any projection geometry without affecting performance, which makes our algorithm competitive with open-source packages. The proposed method uses ray tracing by taking advantage of a computer-graphics library. In the context of image reconstruction, coarse- and fine-grid regimes have both been studied. We obtain convincing experimental results for continuous-domain inverse problems, where basis functions of higher-degree lead to better reconstructions.  From our experiments, we recommend the box-spline of degree 2 as the basis that provides the best tradeoff between runtime and quality. Our methodology can be naturally extended to 3D settings, which we leave for future work. The code of our implementation\footnote{https://github.com/HaouchatY/rt\_tomo\_basis\_functions/} will be made publicly available along with the scripts that reproduce the figures of the paper.

\section{Acknowledgments}
The research leading to these results has received funding from the European Research Council under Grant ERC-2020-AdG FunLearn-101020573 and by the Swiss National Science
Foundation under Sinergia Grant CRSII5\_198569. \fixme{We acknowledge the ENAC Interdisciplinary Platform for X-ray micro-tomography (PIXE) for providing the real CT data used in our experiments.}

\newpage
\appendix

\label{appendixproof}
\begin{proof}[Proof of Lemma~\ref{lemma}]
We provide a proof of (1) and note that (2) can be proven in the same manner.\\

Let $(y, \theta) \in \mathbb{R}^2$ represent the parameters of a ``mainly vertical'' ray. Consider a basis function $\varphi(\cdot - \mathbf{k})$ that contributes to the integral, for some $\mathbf{k} \in \Omega$. Without loss of generality, we can assume that $\mathbf{k} = \mathbf{0}$, up to a shift in the parameter \( y \) to \(( y - \langle \mathbf{k}, \V\theta^\perp \rangle) \). 
Since $\varphi$ is compactly supported inside a disk of radius $R$, it holds that
\begin{equation}
\label{bounded}
    \forall \mathbf{x}\in\R^2 \ \text{s.t.}\ \left\|\mathbf{x}\right\|_2\geq R,\ \ \varphi(\mathbf{x}) =0.
\end{equation}
Since the ray is ``mainly vertical,'' it intersects the horizontal axis directed by $\mathbf{e}_1$ at the horizontal component  
\begin{equation}
\gamma(y, \theta) = y/\sin\theta,    
\end{equation}

\noindent where $|\sin\theta| \geq 1/\sqrt{2}$. From (\ref{rayparam}), we note that $y$ represents the Euclidean distance from the ray to the origin. It follows that  
\begin{equation}
|\gamma(y, \theta)| \leq |\sqrt{2} y| \leq \sqrt{2} R.
\end{equation}
The last inequality holds because, if the ray crosses the support of $\varphi$, then $y$ must be at most $R$, as stated in \eqref{bounded}. Now, suppose that the domain is discretized into a uniform grid with basis vectors $(\mathbf{e}_1, \mathbf{e}_2)$ and stepsize 1. Then, the index along $\mathbf{e}_1$ of the cell containing the intersection of the ray with the horizontal axis is bounded by
$\lceil \sqrt{2} R - 1 \rceil$.
Consequently, the ray will cross one index in the set  
\begin{equation}
\{\mathbf{k} - \lceil \sqrt{2} R - 1 \rceil \mathbf{e}_1, \dotsc, \mathbf{k} + \lceil \sqrt{2} R - 1 \rceil \mathbf{e}_1\}.
\end{equation}
Conversely, when we evaluate the contribution of the cells up to $\lceil \sqrt{2} R - 1 \rceil$ horizontal neighbors away from those directly crossed by the ray, the basis function centered at $\mathbf{k}$ will necessarily be included in the computation.  
\end{proof}

\begin{proof}[Proof of Proposition~\ref{prop}]
    Let $\varphi$ be a function in $L^1(\R^2)$ that verifies (\ref{prod}).
    The Fourier-slice theorem \cite{fourierslice} states that the Fourier transform of the x-ray projection in a direction $\V\theta$ of $\varphi$ equals the $1$-dimensional slice through the origin of $\hat{\varphi}$ in the same direction. This is formalized as
    \begin{equation}
    \label{slice}
        \begin{aligned}[b]
            \forall \xi\in\R,\ \ \hat{\varphi_\theta} (\xi) &= \hat{\varphi}(\xi \V\theta)\\
        &= \prod_{d=1}^D \hat{\phi}_d(\xi \langle \V\theta, \mathbf{u}_d \rangle),
        \end{aligned}
    \end{equation}

\noindent where $\varphi_\theta$ is as in (\ref{eqproj}). Here, each term of the Fourier factorization is a dilation of $\hat{\phi}_d$ by a factor $\langle \V\theta, \mathbf{u}_d \rangle$. When $\langle \V\theta, \mathbf{u}_d \rangle=0$, the $d$th term is the constant $\hat{\phi}_d(0)$.
Otherwise, the dilation property of the Fourier transform states that
\begin{equation}
    \label{dil}
    \forall \xi \in \R, \quad \hat{f}\Big(\lambda \xi \Big) = \dfrac{1}{|\lambda|}
    \widehat{f_{\lambda^{-1}}}(\xi),
\end{equation}
where $f_{\lambda^{-1}}(x) := f\Big(\dfrac{x}{\lambda}\Big)$ with $\lambda \neq 0$.\\

\noindent The proof is finally completed by injecting (\ref{dil}) into (\ref{slice}) with $f=\phi_d$. We conclude using the Fourier-convolution theorem.
\end{proof}

\begin{table}[t]
\caption{PSNR and SSIM in terms of model (separable B-splines of degree 0, 1, and 2) and grid size for cone-beam reconstructions. Bold: best PSNR and SSIM in each row.  }
    \normalsize
    \centering
\begin{tabular}{rccc}
    \toprule
    {$N_{\text{down}}$} & B-splines $0$ & B-splines $1$ & B-splines $2$ \\
    \midrule
    50    & (18.29, 0.54) & (\textbf{21.12}, 0.70) & (20.50, \textbf{0.74}) \\
    100   & (23.21, 0.71) & (25.80, 0.84) & \textbf{(26.05, 0.87)} \\
    150   & (26.01, 0.81) & (28.27, 0.91) & \textbf{(28.63, 0.93)} \\
    250   & (29.78, 0.89) & (32.64, 0.96) & \textbf{(32.87, 0.97)} \\
    375   & (32.92, 0.93) & (35.86, 0.98) & \textbf{(35.94, 0.98)} \\
    1000  & (38.74, 0.98) & (41.10, 0.99) & \textbf{(41.34, 0.99)} \\
    \bottomrule
\end{tabular}
    
    \label{tab:psnr_ssim_cone_separable}
\end{table}

\begin{figure}[t]
\centering
         \begin{tikzpicture}
         \node[anchor=south west,inner sep=0] (main2) at (0,0.5) {\includegraphics[scale = 0.85]{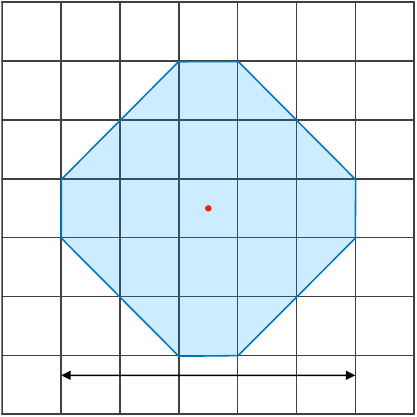}};
         \node[anchor=center, text=black, font=\small ] at (3cm, 0.8cm) {\contour{white}{$L=2K+1$}};
         \end{tikzpicture}
\caption{Octagonal shape of size $L$ used in Lemma~\ref{lemma}. It corresponds to the largest support for $K$ neighbor evaluations.}     
\label{octagon_shape}
\end{figure}

\printbibliography

@article{Jakob2020DrJit, 
author = {Jakob, W. and Speierer, S. and Roussel, N. and Vicini, D.}, 
title = {Dr.Jit: a just-in-time compiler for differentiable rendering}, 
year = {2022}, 
issue_date = {July 2022}, 
publisher = {Association for Computing Machinery}, 
address = {New York, NY, USA}, 
volume = {41}, 
number = {4}, 
issn = {0730-0301}, 
url = {https://doi.org/10.1145/3528223.3530099}, 
doi = {10.1145/3528223.3530099}, 
journal = {ACM Trans. Graph.}, 
month = jul, 
articleno = {124}, 
numpages = {19} }

@ARTICLE{fesslercone,
  author={Long, Y. and Fessler, J. A. and Balter, J. M.},
  journal={IEEE Transactions on Medical Imaging}, 
  title={3D Forward and Back-Projection for X-Ray CT Using Separable Footprints}, 
  year={2010},
  month=jun,
  volume={29},
  number={11},
  pages={1839-1850},
  doi={10.1109/TMI.2010.2050898}}

@book{natterer2001mathematics,
  title={The Mathematics of Computerized Tomography},
  author={Natterer, F.},
  year={2001},
  publisher={SIAM}
}

@article{Kak,
author = {Kak, A. C. and Slaney, M. and Wang, G.},
title = {Principles of computerized tomographic imaging},
journal = {Medical Physics},
volume = {29},
number = {1},
pages = {107-107},
doi = {https://doi.org/10.1118/1.1455742},
year = {2002},
month = jan,
}

@book{epstein2008,
author = {Epstein, C. L.},
title = {Introduction to the Mathematics of Medical Imaging},
publisher = {Society for Industrial and Applied Mathematics},
year = {2007},
doi = {10.1137/9780898717792},
address = {Philadelphia, PA},
edition   = {2},
URL = {https://epubs.siam.org/doi/abs/10.1137/9780898717792},
eprint = {https://epubs.siam.org/doi/pdf/10.1137/9780898717792}
}

@article{siddon,
author = {Siddon, R. L.},
title = {Fast calculation of the exact radiological path for a three-dimensional CT array},
journal = {Medical Physics},
volume = {12},
number = {2},
pages = {252-255},
doi = {https://doi.org/10.1118/1.595715},
year = {1985}
}

@INPROCEEDINGS{rt,
  author={Zhao, H. and Reader, A.J.},
  booktitle={2003 IEEE Nuclear Science Symposium.}, 
  title={Fast ray-tracing technique to calculate line integral paths in voxel arrays}, 
  year={2003},
  month=dec,
  number={},
  pages={2808-2812 Vol.4},
  doi={10.1109/NSSMIC.2003.1352469}}

@article{jacobs1998fast,
  title={A fast algorithm to calculate the exact radiological path through a pixel or voxel space},
  author={Jacobs, F. and Sundermann, E. and De Sutter, B. and Christiaens, K. and Lemahieu, I.},
  journal={Journal of Computing and Information Technology},
  volume={6},
  number={1},
  pages={89--94},
  year={1998},
  publisher={Sveu{\v{c}}ili{\v{s}}te u Zagrebu Sveu{\v{c}}ili{\v{s}}ni ra{\v{c}}unski centar}
}

@ARTICLE{unserblu,
  author={Unser, M. and Blu, T.},
  journal={IEEE Transactions on Signal Processing}, 
  title={Wavelet theory demystified}, 
  year={2003},
  volume={51},
  number={2},
  pages={470-483},
  doi={10.1109/TSP.2002.807000}}

@book{herman2009fundamentals,
  title={Fundamentals of computerized tomography: image reconstruction from projections},
  author={Herman, G. T.},
  year={2009},
  month=jul,
  publisher={Springer Science \& Business Media}
}

@article{fourierslice, 
author = {Ng, R.}, 
title = {Fourier slice photography}, 
year = {2005}, issue_date = {July 2005}, 
publisher = {Association for Computing Machinery}, 
address = {New York, NY, USA}, 
volume = {24}, 
number = {3}, 
issn = {0730-0301}, 
url = {https://doi.org/10.1145/1073204.1073256}, 
doi = {10.1145/1073204.1073256}, 
journal = {ACM Trans. Graph.}, 
month = jul, 
pages = {735–744}, 
numpages = {10} }

@article{unser1993b,
  title={B-spline signal processing. I. Theory},
  author={Unser, M. and Aldroubi, A. and Eden, M.},
  journal={IEEE Transactions on Signal Processing},
  volume={41},
  number={2},
  pages={821--833},
  year={1993},
  month=feb,
  publisher={IEEE}
}

@ARTICLE{entezari,
  author={Entezari, A. and Nilchian, M. and Unser, M.},
  journal={IEEE Transactions on Medical Imaging}, 
  title={A Box Spline Calculus for the Discretization of Computed Tomography Reconstruction Problems}, 
  month = {Aug.},
  year={2012},
  volume={31},
  number={8},
  pages={1532-1541},
  doi={10.1109/TMI.2012.2191417}}

@INPROCEEDINGS{mehrsahtv,
  author={Pourya, M. and Haouchat, Y. and Unser, M.},
  booktitle={2024 IEEE International Symposium on Biomedical Imaging (ISBI)}, 
  title={A continuous-domain solution for computed tomography with hessian total-variation regularization}, 
  year={27-30 May 2024},
  volume={},
  number={},
  pages={1-5},
  doi={10.1109/ISBI56570.2024.10635730},
  location = {Athens, Greece},
}

@ARTICLE{mehrsamri,
  author={Pourya, M. and Boquet-Pujadas, A. and Unser, M.},
  journal={IEEE Transactions on Computational Imaging}, 
  title={A box-spline framework for inverse problems with continuous-domain sparsity constraints}, 
  year={2024},
  month=may,
  volume={10},
  number={},
  pages={790-805},
  doi={10.1109/TCI.2024.3402376}}

@ARTICLE{Bresenham,
  author={Bresenham, J. E.},
  journal={IBM Systems Journal}, 
  title={Algorithm for computer control of a digital plotter}, 
  year={1965},
  month=dec,
  volume={4},
  number={1},
  pages={25-30},
  doi={10.1147/sj.41.0025}}

@article{astra,
  title = {Fast and flexible X-ray tomography using the ASTRA toolbox},
  volume = {24},
  ISSN = {1094-4087},
  url = {http://dx.doi.org/10.1364/OE.24.025129},
  DOI = {10.1364/oe.24.025129},
  number = {22},
  journal = {Optics Express},
  publisher = {Optica Publishing Group},
  author = {van Aarle,  W. and Palenstijn,  W. J. and Cant, J. and Janssens, E. and Bleichrodt,  F. and Dabravolski, A. and De Beenhouwer,  J. and Joost Batenburg,  K. and Sijbers,  J.},
  year = {2016},
  month = oct,
  pages = {25129}
}

@inproceedings{rtk,
  title = {The Reconstruction Toolkit (RTK),  an open-source cone-beam CT reconstruction toolkit based on the Insight Toolkit (ITK)},
  ISSN = {1742-6596},
  url = {http://dx.doi.org/10.1088/1742-6596/489/1/012079},
  DOI = {10.1088/1742-6596/489/1/012079},
  booktitle = {Journal of Physics: Conference Series},
  author = {Rit,  S. and Vila Oliva,  M. and Brousmiche,  S. and Labarbe,  R. and Sarrut,  D. and Sharp,  G. C.},
  year = {6–9 May 2013},
  location = {Melbourne, Australia},
  pages = {012079}
}

@article{tigre,
  title = {TIGRE: A MATLAB-GPU toolbox for CBCT image reconstruction},
  volume = {2},
  ISSN = {2057-1976},
  url = {http://dx.doi.org/10.1088/2057-1976/2/5/055010},
  DOI = {10.1088/2057-1976/2/5/055010},
  number = {5},
  journal = {Biomedical Physics \& Engineering Express},
  publisher = {IOP Publishing},
  author = {Biguri,  A. and Dosanjh, M. and Hancock, S. and Soleimani, M.},
  year = {2016},
  month = sep,
  pages = {055010}
}

@inproceedings{siddonaccelTigre,
  series = {NSSMIC-99},
  title = {A fast ray-tracing technique for TCT and ECT studies},
  volume = {3},
  url = {http://dx.doi.org/10.1109/NSSMIC.1999.842846},
  DOI = {10.1109/nssmic.1999.842846},
  booktitle = {IEEE Nuclear Science Symposium Conference Record. Nuclear Science Symposium and Medical Imaging Conference},
  year = {24-30 Oct. 1999},
  author = {Han,  G. and Liang,  Z. and You,  J.},
  pages = {1515–1518},
  collection = {NSSMIC-99},
  location = {Seattle, WA, USA},
}

@article{josephkernel,
  title = {An improved algorithm for reprojecting rays through pixel images},
  volume = {1},
  ISSN = {1558-254X},
  url = {http://dx.doi.org/10.1109/TMI.1982.4307572},
  DOI = {10.1109/tmi.1982.4307572},
  number = {3},
  journal = {IEEE Transactions on Medical Imaging},
  publisher = {Institute of Electrical and Electronics Engineers (IEEE)},
  author = {Joseph,  P. M.},
  year = {1982},
  month = nov,
  pages = {192–196}
}

@inproceedings{distancedriven2,
  series = {NSSMIC-02},
  title = {Distance-driven projection and backprojection},
  volume = {3},
  url = {http://dx.doi.org/10.1109/NSSMIC.2002.1239600},
  DOI = {10.1109/nssmic.2002.1239600},
  booktitle = {IEEE Nuclear Science Symposium Conference Record},
  author = {De Man,  B. and Basu,  S.},
  pages = {1477–1480},
  collection = {NSSMIC-02},
  year = {10-16 Nov. 2002},
  location = {Norfolk, VA, USA},

}

@inproceedings{convolutional2,
  title = {A convolutional framework for forward and back-projection in fan-beam geometry},
  url = {http://dx.doi.org/10.1109/ISBI.2019.8759285},
  DOI = {10.1109/isbi.2019.8759285},
  booktitle = {2019 IEEE 16th International Symposium on Biomedical Imaging (ISBI 2019)},
  author = {Zhang,  K. and Entezari,  A.},
  year = {8-11 April 2019},
  pages = {1455–1458},
  location ={Venice, Italy},
}

@article{Brooks1978,
  title = {A new approach to interpolation in computed tomography},
  volume = {2},
  ISSN = {0363-8715},
  url = {http://dx.doi.org/10.1097/00004728-197811000-00010},
  DOI = {10.1097/00004728-197811000-00010},
  number = {5},
  journal = {Journal of Computer Assisted Tomography},
  publisher = {Ovid Technologies (Wolters Kluwer Health)},
  author = {Brooks,  R. A. and Weiss,  G. H. and Talbert,  A. J.},
  year = {1978},
  month = nov,
  pages = {577–585}
}

@article{horbelt,
  title = {Discretization of the Radon transform and of its inverse by spline convolutions},
  volume = {21},
  ISSN = {0278-0062},
  url = {http://dx.doi.org/10.1109/TMI.2002.1000260},
  DOI = {10.1109/tmi.2002.1000260},
  number = {4},
  journal = {IEEE Transactions on Medical Imaging},
  publisher = {Institute of Electrical and Electronics Engineers (IEEE)},
  author = {Horbelt,  S. and Liebling,  M. and Unser,  M.},
  year = {2002},
  month = apr,
  pages = {363–376}
}

@misc{dataset,
  doi = {10.7937/K9/TCIA.2016.JGNIHEP5},
  url = {https://www.cancerimagingarchive.net/collection/tcga-luad/},
  author = {Albertina,  B. and Watson,  M. and Holback,  C. and Jarosz,  R. and Kirk,  S. and Lee,  Y. and Rieger-Christ,  K. and Lemmerman,  J.},
  title = {The Cancer Genome Atlas Lung Adenocarcinoma Collection (TCGA-LUAD)},
  publisher = {The Cancer Imaging Archive},
  year = {2016},
  copyright = {Creative Commons Attribution 3.0 Unported}
}

@ARTICLE{rattey,
  author={Rattey, P. and Lindgren, A.},
  journal={IEEE Transactions on Acoustics, Speech, and Signal Processing}, 
  title={Sampling the 2-D Radon transform}, 
  year={1981},
  month=oct,
  volume={29},
  number={5},
  pages={994-1002},
  doi={10.1109/TASSP.1981.1163686}}

@ARTICLE{fessler,
  author={Fessler, J.A.},
  journal={IEEE Transactions on Medical Imaging}, 
  title={Penalized weighted least-squares image reconstruction for positron emission tomography}, 
  year={1994},
  month=jun,
  volume={13},
  number={2},
  pages={290-300},
  doi={10.1109/42.293921}}

@ARTICLE{pet_aleix,
  author={Boquet-Pujadas, A. and Saidi, J. and Vicente, M. and Paolozzi, L. and Dong, J. and del Aguila Pla, P. and Iacobucci, G. and Unser, M.},
  journal={IEEE Transactions on Radiation and Plasma Medical Sciences}, 
  title={A Silicon-Pixel Paradigm for PET}, 
  year={2025},
  volume={9},
  number={2},
  pages={228-246},
  doi={10.1109/TRPMS.2024.3456241}}
\end{document}